\documentclass[11pt,a4paper]{article}

\usepackage[utf8]{inputenc}
\usepackage[T1]{fontenc}
\usepackage{amsfonts}
\usepackage{graphicx}
\usepackage{mathtools}
\usepackage{enumitem}
\usepackage{booktabs}
\usepackage{etoolbox}
\usepackage[a4paper,top=20mm,bottom=30mm,left=20mm,right=20mm]{geometry}
\usepackage{amssymb}
\usepackage{amsthm}
\usepackage{xcolor}
\usepackage{microtype}
\usepackage{setspace}
\usepackage[hidelinks]{hyperref}

\numberwithin{equation}{section}
\newtheorem{theorem}{Theorem}[section]
\newtheorem{proposition}[theorem]{Proposition}

\newtheorem{lemma}[theorem]{Lemma}
\theoremstyle{definition}
\newtheorem{assumption}[theorem]{Assumption}
\newtheorem{remark}[theorem]{Remark}
\newtheorem{example}[theorem]{Example}

\newcommand{\email}[1]{\href{mailto:#1}{\nolinkurl{#1}}}

\AtEndEnvironment{assumption}{\hfill\(\diamond\)}
\AtEndEnvironment{remark}{\hfill\(\diamond\)}
\AtEndEnvironment{example}{\hfill\(\diamond\)}

\newcommand{\beq}{\begin{equation}}
\newcommand{\eeq}{\end{equation}}
\newcommand{\beqa}{\begin{eqnarray}}
\newcommand{\eeqa}{\end{eqnarray}}
\newcommand{\nn}{\nonumber}

\newcommand{\mE}{{\mathcal E}}
\newcommand{\mF}{{\mathcal F}}

\newcommand{\mH}{{\mathcal H}}

\newcommand{\mbZ}{\mathbb{Z}}

\newcommand{\mbE}{\mathbb{E}} 
\newcommand{\mbR}{\mathbb{R}} 
 
\newcommand{\mbC}{\mathbb{C}} 

\newcommand{\sB}{{\sf B}}
\title{Stochastic stability of master--slave synchronization for a class of PDEs with Burgers-type nonlinearity\thanks{Submitted to the editors \today.
This work has been partially supported by grant PID2024-158181NB-I00 NISA, funded by MCIN/AEI/10.13039/501100011033, and by the Community of Madrid under grant IDEA-CM (ref. TEC-2024/COM-89).}}

\author{Joaqu\'in Miguez\thanks{Universidad Carlos III de Madrid,
Avenida de la Universidad 30, 28911 Legan\'es (Madrid), Spain
(\email{joaquin.miguez@uc3m.es}).}
\and In\'es P. Mari\~no\thanks{Universidad Rey Juan Carlos,
Calle Tulip\'an s/n, 28933 M\'ostoles (Madrid), Spain
(\email{ines.perez@urjc.es}).}}

\hypersetup{
  pdftitle={Stochastic stability of master-slave synchronization for a class of PDEs with Burgers-type nonlinearity},
  pdfauthor={Joaquin Miguez and Ines P. Marino}
}
\date{\today}

\begin{document}

\maketitle

\begin{abstract}
Master--slave synchronization reconstructs the evolving state of a reference (or master) system by continuously correcting a second system, the slave, using observations of the master. In applications, the observations are often contaminated by noise, so the corrections injected into the slave are themselves uncertain. In this paper we aim at quantifying the effect of noisy observations on master--slave synchronization for a family of evolution equations that combine a quadratic transport term with a polynomial linear differential operator. The family includes the Burgers, Kuramoto--Sivashinsky, Kawahara, Benney--Lin and Nikolaevskiy equations, among others. Under periodic boundary conditions, a Fourier--Galerkin approximation of order $K$ yields a finite-dimensional master--slave system in which white observation noise of intensity $\sigma$ is modeled through an It\^o stochastic differential equation. In this framework, we first establish local exponential stability of noiseless synchronization for sufficiently strong scalar coupling. Then we compare the stochastic synchronization error with its noiseless counterpart resulting from the same initial slave state, isolating in this way the effect of noise. Assuming bounded reference paths, the mean-square deviation between the noiseless and noisy synchronization errors is $\mathcal O(\sigma^2)$ on finite time intervals. Under additional regularity assumptions, we also bound the error between the finite-resolution stochastic slave model and the infinite-dimensional deterministic master field, explicitly separating spectral approximation error, deterministic synchronization error and observation noise. Numerical experiments for the Kawahara and Nikolaevskiy equations confirm the predicted noise scaling.
\end{abstract}



%
\section{Introduction}

%
\subsection{Motivation and background}
\label{sIntroSync}

Reconstructing the evolving state of a physical system from observations is a common problem in science and engineering. Master--slave synchronization offers a simple approach: observations of a reference system, the \emph{master}, continuously correct a second system, the \emph{slave}, so that it follows the same trajectory. Since the work of Pecora and Carroll~\cite{Pecora90}, this idea has found various applications, e.g., in secure communications, state estimation, and parameter identification~\cite{Pikovsky01,Boccaletti02,Parlitz96,Marinho05,Marinho06b,Asheghan13}. 

For spatially extended systems, master-slave synchronization provides a way to recover an evolving field from measurements at finitely many locations \cite{KocarevTasevParlitz1997}. This is the same principle that underlies modern data assimilation \cite{Duane06,LawStuartZygalakis2015,ReichCotter2015}, where observations guide a computational model of a physical process, for example in weather and climate prediction~\cite{Duane06}. In continuous data assimilation (CDA) by nudging, the discrepancy between observed and predicted quantities is fed back continuously into an auxiliary PDE \cite{AzouaniOlsonTiti2014}. The observed process and the auxiliary model play the roles of master and slave, and the coupling gain determines how strongly the model is pulled towards the data. This viewpoint has proved remarkably flexible: CDA has been analysed and implemented for the Navier--Stokes equations \cite{AzouaniOlsonTiti2014,BessaihOlsonTiti2015}, magnetohydrodynamics \cite{BiswasHudsonLariosPei2018}, Rayleigh--B\'enard convection \cite{FarhatJollyTiti2015}, the surface quasi-geostrophic equation \cite{JollyMartinezTiti2017}, and shell models of turbulence \cite{ChenFarhatLunasin2023}, among others.

Both the synchronization and data assimilation perspectives raise similar practical questions. In any real application the observations are contaminated by measurement error, so the corrections injected into the slave are themselves uncertain. How accurately can the model track the reference in that case? Does the synchronization error remain bounded, or does the persistent injection of noise destabilize the slave? And how does the achievable accuracy depend on the noise intensity, on the coupling gain, and on the resolution at which the slave is represented? These are the questions we address in the paper.

We study them for a family of nonlinear evolution equations arising in fluid dynamics, wave propagation, and pattern formation. Its members include the Burgers equation, a basic model of nonlinear transport and viscous dissipation~\cite{Burgers1948,BecKhanin2007}; the Kuramoto--Sivashinsky equation, associated with flame-front instabilities and reaction--diffusion systems~\cite{Sivashinsky1980,Kuramoto1978,BraunerEtAl2010,AransonKramer2002}; Kawahara-type equations for dispersive waves in fluids and plasmas~\cite{Kawahara1972,KakutaniOno1969,SprengerHoefer2017,ElTantawySalasAlharthi2021}; Benney--Lin equations for waves on liquid films~\cite{Benney1966,Lin1974,CrasterMatar2009}; and the Nikolaevskiy equation, introduced for nonlinear seismic waves and subsequently studied as a model of spatiotemporal chaos~\cite{BeresnevNikolaevskii1993,XiToralGuntonTribelsky2000,SimbawaMatthewsCox2010}. Despite their different physical origins, all of them combine a quadratic transport term with a linear operator built from spatial derivatives. In particular, they can be written as
\begin{equation}
\label{eq:intro-pde}
u_t+c_1\,u\,u_x+P(\partial_x)\,u=0,
\end{equation}
where $u(t,x)$ is a real field, $c_1\in\mathbb R$ and $P$ is a polynomial with constant real coefficients. Subscripts in $u_t$ and $u_x$ denote partial derivatives. The operator $P(\partial_x)$ encodes the linear physics of each model: viscous dissipation in Burgers, fourth-order damping and second-order anti-damping in Kuramoto--Sivashinsky and Nikolaevskiy, and higher-order dispersion in Kawahara. The common structure of \eqref{eq:intro-pde} allows a unified synchronization analysis across physically different models, including purely dispersive ones.

Under periodic boundary conditions, we represent the master and slave by truncated Fourier series and drive the slave with reconstructed observations of the master. Earlier work by the authors considered deterministic synchronization and parameter estimation for the Kuramoto--Sivashinsky equation~\cite{Miguez2024}. Here we extend the analysis to the family~\eqref{eq:intro-pde} and rigorously quantify the effect of noisy measurements. The master remains deterministic; white observation noise enters through the coupling and makes the slave stochastic.

Rigorous results on synchronization with noisy observations are comparatively scarce, particularly outside the parabolic setting. Bessaih, Olson, and Titi~\cite{BessaihOlsonTiti2015} obtained mean-square error bounds for the two-dimensional Navier--Stokes equations with stochastically noisy data, and Olja\v{c}a, Br\"ocker, and Kuna~\cite{OljacaBroeckerKuna2018} derived almost-sure error bounds for dissipative systems with unbounded (Gaussian) observation noise. More recently, Br\"ocker et al.~\cite{BroeckerDelSartoHieberPalmaZoechling2026} developed an infinite-dimensional framework for semilinear parabolic equations with multiplicative observation noise, proving mean-square bounds and, under an additional noise-integrability condition, almost-sure synchronization, while Bessaih et al.~\cite{BessaihFerrarioLandoulsiZanella2026} extended the nudging algorithm to the two-dimensional stochastic Navier--Stokes equations with additive and multiplicative noise. All of these frameworks assume, in one form or another, a coercive or parabolic linear operator. Dispersive models, and more generally equations whose linear part is not sign-definite, fall outside their scope. The analysis below in this paper is intended to cover that gap.

%
\subsection{Contributions}
\label{sIntroContributions}

We consider the family~\eqref{eq:intro-pde} with periodic boundary conditions and a Fourier--Galerkin truncation of order $K$, which reduces the problem to a finite-dimensional system of real dimension $2K+1$. The slave is driven by observations of the master, reconstructed from samples on a spatial grid. We first establish local exponential stability of noiseless synchronization on a finite interval $[0,T]$ for sufficiently strong scalar coupling, and then replace the exact observations by noisy ones, modelled as a finite-dimensional It\^o stochastic differential equation. The main technical device is to compare the stochastic synchronization error $e(t)$ with the noiseless reference error $e_{\rm ref}(t)$ obtained from the same initial slave state, and to study the deviation $\Delta(t):=e(t)-e_{\rm ref}(t)$. This isolates the effect of noise from the deterministic synchronization transient.

For a fixed truncation order $K$, finite time horizon $T$ and observational noise intensity $\sigma$, sufficient resolution of the observation grid, bounded master and reference paths, and sufficiently strong coupling yield an $\mathcal O(\sigma^2)$ mean-square bound on the deviation $\Delta(t)$, together with a supremum bound and a tail-probability estimate (see Proposition~\ref{prop:stochastic-closeness}). A key ingredient is that the quadratic transport remainder cancels exactly in the energy balance, so the estimate holds on the whole interval without restricting the deviation to a neighbourhood of the origin. The required coupling strength may depend on $K$ and $T$.

A time-uniform estimate is obtained under an explicit dissipativity condition, namely a lower bound on $\operatorname{Re}[-P(i\xi)]$ over all frequencies together with time-uniform bounds on the reference slave field and its spatial derivative. This yields a positive dissipativity gap ($\mu_K$) and a finite lower bound on the coupling strength ($d$) necessary for synchronization. If these conditions are met, the pointwise mean-square deviation is then bounded uniformly in time by $\frac{1}{2\mu_K}\sigma^2d^2\|R_K\|^2$, where $R_K$ is the least-squares reconstruction matrix that yields estimates of the master Fourier modes from the noisy observations (see Proposition~\ref{prop:global-stochastic-closeness}). Parseval's identity transfers both spectral estimates directly to the truncated physical fields.


Finally, we bound the error $\mE_\infty^{(K)}=u-v_K$ between the full master field and the order-$K$ slave. Under Sobolev regularity, reconstruction and dissipativity assumptions uniform in $K$, this error splits into four terms: the unresolved Fourier tail, the Fourier--Galerkin error, the deterministic reference error and the stochastic deviation. The first two vanish as $K\to\infty$ at each fixed time, and additional time-uniform assumptions allow a subsequent long-time limit, leaving an $\mathcal O(\sigma^2)$ mean-square residual (see Proposition~\ref{prop:physical-error-uniform-K}). Numerical experiments with the Kawahara and Nikolaevskiy equations illustrate the predicted noise scaling and spectral error.

Compared with the existing literature, the stability estimates at fixed $K$ require no parabolic coercivity of the linear operator, and therefore apply to purely dispersive models such as the Kawahara equation. The formulation in Fourier space makes the dependence on the observation reconstruction matrix $R_K$ explicit, so that its (Hilbert--Schmidt) norm enters the error bounds directly. Finally, the coupling condition (lower bound on $d$) is stated in terms of the linear symbol of the PDE, which is computable for any member of the family~\eqref{eq:intro-pde}.

\subsection{Outline of the paper}
\label{sIntroOrganization}

Section~\ref{sPDEFamily} introduces the PDE family and its Fourier representation. Section~\ref{sMS} studies deterministic synchronization, and Section~\ref{sStochasticSynch} constructs the stochastic slave from noisy observations. Section~\ref{sec:stochastic-synchronization} develops the stochastic and physical-space error bounds. Section~\ref{sec:numerical-results} presents numerical experiments, and Section~\ref{sConclusions} gives conclusions and future directions. Proofs are collected in the appendices.

%
\section{A class of nonlinear PDEs and their Fourier representation}
\label{sPDEFamily}

%
\subsection{A family of PDEs with Burgers-like nonlinearity}

Let \(u\equiv u(t,x)\) be a real-valued scalar field, with \(t\in[0,T]\) and \(x\in\mathbb R\), that evolves according to equation \eqref{eq:intro-pde}. The linear operator \(P(\partial_x)\) in \eqref{eq:intro-pde} can be explicitly written as
\begin{equation}
\label{eq:operator-polynomial}
P(\partial_x)
=
c_0\partial_x+c_2\partial_x^2+c_3\partial_x^3
+c_4\partial_x^4+c_5\partial_x^5+c_6\partial_x^6 ,
\end{equation}
where $P(0)=0$ and \(c_0, c_2, \ldots, c_6 \in\mathbb R\) are constant coefficients, \(\partial_x=\partial/\partial x\) and $\partial_x^n$ denotes the $n$-th order partial spatial derivative. 
The common structure to this class of models is the conservative nonlinear self-transport term
\(uu_x=\partial_x(u^2/2)\), together with linear effects that may be dissipative, antidissipative, dispersive or regularizing depending on the coefficients in \(P\). This class contains lower-order transport-diffusion models as well as several higher-order pattern-forming equations. The analysis could be carried out for a polynomial of arbitrary finite degree $n$, however $n=6$ already includes the existing examples of interest.

\begin{example}
Fix $c_1=1$ for all examples. The viscous Burgers equation \cite{Burgers1948},
\begin{equation}
\label{eq:burgers}
u_t+u u_x=\nu u_{xx}, \quad \nu>0,
\end{equation}
is recovered from \eqref{eq:intro-pde} by taking \(c_2=-\nu\) and $c_0=c_3=c_4=c_5=c_6=0$, i.e., \(P_{\rm B}(\partial_x)=-\nu\partial_x^2\). The generalized Kuramoto--Sivashinsky equation \cite{Kuramoto1976,Kuramoto1978,Sivashinsky1977,Sivashinsky1980,Miguez2024},
\begin{equation}
\label{eq:ks}
u_t+u u_x+\alpha u_{xx}+\beta u_{xxx}+\gamma u_{xxxx}=0,
\end{equation}
is obtained with \(c_2=\alpha\), \(c_3=\beta\), \(c_4=\gamma\), and \(c_0=c_5=c_6=0\). A Nikolaevskiy-type equation with dispersive terms \cite{BeresnevNikolaevskii1993,XiToralGuntonTribelsky2000,CoxMatthews2007,SimbawaMatthewsCox2010} can be written as
\begin{equation}
\label{eq:nikolaevskiy-compact}
u_t
=
-\partial_x^2\!\left(r-(1+\partial_x^2)^2\right)u
-u u_x+\alpha u_{xxx}+\beta u_{xxxxx},
\end{equation}
which expands to
\begin{equation}
\label{eq:nikolaevskiy-expanded}
u_t+u u_x-(1-r)u_{xx}-\alpha u_{xxx}
-2u_{xxxx}-\beta u_{xxxxx}-u_{xxxxxx}=0.
\end{equation}
Thus it corresponds to \(c_2=-(1-r)\), \(c_3=-\alpha\), \(c_4=-2\), \(c_5=-\beta\), \(c_6=-1\) and \(c_0=0\). Kawahara-type equations \cite{Kawahara1972,Kawahara1978,KakutaniOno1969,BiagioniLinares1997} and Benney--Lin equations for thin-film and long-wave instabilities \cite{Benney1966,Lin1974,OronDavisBankoff1997,CrasterMatar2009} are obtained in the same way by selecting the appropriate third-, fourth-, and fifth-order linear coefficients in \eqref{eq:operator-polynomial}.
\end{example}

%
\subsection{Fourier representation}

Assume periodic boundary conditions with spatial period \(X>0\) (namely $u(t,x+X)=u(t,x)$, where $x\in\mathbb R$), let $\omega_0=\frac{2\pi}{X}$ be the fundamental (spatial) frequency and define the Fourier basis functions 
\begin{equation}
\phi_k(x):={\rm e}^{i\omega_0 kx}, \quad k\in\mathbb Z.
\nonumber
\end{equation}
If the field $u(t,x)$ is regular enough\footnote{It is sufficient to assume that $u\in C([0,T];L^2_{\rm per}(0,X))$ and \eqref{eq:intro-pde} is satisfied in the sense of distributions, with $uu_x$ interpreted as $\frac12\partial_x(u^2)$} then, for every $t\in[0,T]$, we have the Fourier expansion,
convergent in $L^2(0,X)$,
\begin{equation}
\label{eq:fourier-expansion}
u(t,x)=\sum_{k\in\mathbb Z}a_k(t)\phi_k(x),
\end{equation}
with Hermitian symmetry, $a_{-k}(t)=a_k(t)^*$, because \(u\) is real-valued. Each $\phi_k$ is an eigenfunction of the linear operator $P(\partial_x)$, which therefore acts diagonally on the Fourier coefficients. Specifically,
\begin{equation}
\label{eq:symbol}
P(\partial_x)\phi_k(x)=P(i\omega_0 k)\phi_k(x),
\end{equation}
where
\begin{equation}
\label{eq:symbol-expanded}
P(i\omega_0 k)
=ic_0\omega_0 k-c_2\omega_0^2k^2-ic_3\omega_0^3k^3
+c_4\omega_0^4k^4+ic_5\omega_0^5k^5-c_6\omega_0^6k^6.
\end{equation}
For the nonlinear term, the $k$-th Fourier coefficient of \(uu_x\) is
\begin{equation}
\label{eq:convolution}
\widehat{(uu_x)}_k
:= \frac{1}{2}\widehat{\partial_x\left(u^2\right)_k}
=\frac{i\omega_0 k}{2}\sum_{\ell\in\mathbb Z}a_\ell a_{k-\ell}.
\end{equation}
Substituting \eqref{eq:symbol} and \eqref{eq:convolution} into \eqref{eq:intro-pde} yields the system of ordinary differential equations (ODEs)
\begin{equation}
\label{eq:modal-ode}
\dot a_k
=
-P(i\omega_0 k)a_k
-\frac{ic_1\omega_0 k}{2}\sum_{\ell\in\mathbb Z}a_\ell a_{k-\ell},
\qquad k\in\mathbb Z .
\end{equation}
This is the common Fourier representation of the family \eqref{eq:intro-pde}, where only the scalar multiplier \(P(i\omega_0 k)\) changes from one model to another.

\begin{example}
For Burgers' equation \eqref{eq:burgers} one readily obtains
\begin{equation}
\label{eq:burgers-modal}
\dot a_k
=
-\nu\omega_0^2 k^2 a_k
-\frac{i\omega_0 k}{2}\sum_{\ell\in\mathbb Z}a_\ell a_{k-\ell},
\end{equation}
while the Nikolaevskiy equation \eqref{eq:nikolaevskiy-expanded} yields
\begin{equation}
\label{eq:nik-modal}
\dot a_k
=
\left(-(1-r)\omega_0^2k^2-i\alpha\omega_0^3k^3
+2\omega_0^4k^4+i\beta\omega_0^5k^5-\omega_0^6k^6\right)a_k
-\frac{i\omega_0 k}{2}\sum_{\ell\in\mathbb Z}a_\ell a_{k-\ell}.
\end{equation}
The generalized Kuramoto--Sivashinsky, Kawahara, and Benney--Lin Fourier systems are obtained in the same manner by substituting their corresponding polynomials into \eqref{eq:modal-ode}.
\end{example}

For a practical computational representation, the Fourier series \eqref{eq:fourier-expansion} can be truncated. For \(K\ge1\), we define the set of retained Fourier modes
$
\{-K,-K+1,\ldots,K-1,K\}
$, which yields the truncated field
\begin{equation}
\label{eq:master-truncated}
u_K(t,x)=\sum_{k=-K}^{K} \bar a_k(t)\phi_k(x),
\quad \text{with} \quad
\bar a_{-k}(t)=\bar a_k^*(t).
\end{equation}
The coefficients $\bar a_k(t)$ satisfy the system of ODEs
\begin{equation}
\label{eq:master-truncated-ode}
\dot{\bar a}_k
=
-P(i\omega_0 k) \bar a_k
-\frac{ic_1\omega_0 k}{2}\sum_{\ell=-K}^{K} \bar a_\ell \bar a_{k-\ell},
\quad \text{with} \quad |k|\le K,
\end{equation}
with the convention \(\bar a_j=0\) whenever \(|j|>K\). Since \(u_K\) is real-valued, one may store only the nonnegative retained modes \(k=0,\ldots,K\), with the negative modes recovered by Hermitian symmetry. 

The accuracy of the Fourier projection is governed by the decay of the Fourier coefficients, which is a standard topic in spectral approximation theory \cite{Canuto06,ShenTangWang2011,Trefethen00}. Let ${\rm P}_K$ denote the $L^2(0,X)$-orthogonal projection onto the Fourier subspace spanned by $\{\phi_k:|k|\le K\}$. If, for a fixed time \(t\), the field satisfies \(u(t,\cdot)\in H^s_{\rm per}(0,X)\) with $s \ge 0$ (where $H^s_{\rm per}(0,X)$ is the Sobolev space of $X$-periodic functions $w(x)=\sum_{k\in\mathbb Z}\widehat w_k\phi_k(x)$ such that $X\sum_{k\in\mathbb Z}(1+\omega_0^2k^2)^s|\widehat w_k|^2<\infty$) then Parseval's identity implies an algebraic tail estimate of the form \cite[Theorem~2.1]{ShenTangWang2011}
\begin{equation}
\label{eq:fourier-tail-sobolev}
\|u(t,\cdot)-{\rm P}_Ku(t,\cdot)\|_{L^2(0,X)}
\le C_s K^{-s}\|u(t,\cdot)\|_{H^s_{\rm per}(0,X)},
\end{equation}
for some $C_s<\infty$. By contrast, if \(u(t,\cdot)\) is analytic in a complex strip around the real axis, then the Fourier coefficients decay exponentially and one obtains \cite[Section~7]{TrefethenWeideman2014}
\begin{equation}
\label{eq:fourier-tail-analytic}
\|u(t,\cdot)-{\rm P}_Ku(t,\cdot)\|_{L^2(0,X)}
\le C e^{-\lambda K}
\end{equation}
for some constants \(C,\lambda>0\) depending on the analyticity radius and on the corresponding analytic norm. This algebraic-versus-exponential dichotomy is classical in Fourier spectral approximation; see, e.g., \cite{Trefethen00,Canuto06}. 
The projected field ${\rm P}_Ku$ and the solution $u_K$ of the closed Fourier--Galerkin system \eqref{eq:master-truncated-ode} are distinct in general: the latter replaces the full nonlinear convolution by the interaction among retained modes. Hence the total master approximation error contains both the unresolved tail $u-{\rm P}_Ku$ and the resolved-mode Galerkin error ${\rm P}_Ku-u_K$. Section~\ref{sec:physical-error} treats these contributions separately.

%
\section{Local synchronization in a master--slave system} \label{sMS}

This section introduces the deterministic master--slave system associated with the truncated Fourier representation above. We first specify the finite-dimensional master and slave dynamics, then derive the synchronization-error equation and state a local synchronization result. 

%
\subsection{Master--slave model}

%

The truncated, real master field $u_K(t,x)$ in \eqref{eq:master-truncated} can be characterized by the Fourier coefficients $\bar a_k(t)$, $|k| \le K$, that evolve according to \eqref{eq:master-truncated-ode} and satisfy $\bar a_{-k}(t) = \bar a_k^*(t)$. Because of the Hermitian symmetry, $u_K(t,x)$ can be uniquely represented by the one-sided vector
\beq
\bar a(t):=[\bar a_0(t),\ldots,\bar a_K(t)]^\top,
\label{eq-def-bar-a}
\eeq
where $\bar a_0(t) \in \mbR$ and $\bar a_k(t) \in \mbC$ for $1 \le k \le K$.

\begin{remark}
The one-sided coefficient vector of a real field belongs to the real vector space
\begin{equation}
\label{eq:real-Fourier-space}
\mathcal H_K:=\{z=(z_0,\ldots,z_K)^\top:z_0\in\mathbb R,\ z_k\in\mathbb C,\ 1\le k\le K\}.
\end{equation}
We note that $\mH_K$ is a vector space over $\mbR$ (not over $\mbC$) because $\alpha z \in \mH_K$ for any $z \in \mH_K$ and $\alpha \in \mbR$ (while, in general, $\beta z \notin \mH_K$ when $\beta \in \mbC$). Any vector $z \in \mH_K$ can be equivalently written as
\beq
\bigl(z_0,\operatorname{Re}z_1,\operatorname{Im}z_1,\ldots,
\operatorname{Re}z_K,\operatorname{Im}z_K\bigr)^\top\in\mbR^{2K+1},
\label{eq-2K+1}
\eeq
hence ${\rm dim}_{\mbR}(\mH_K) = 1+2K$. 

We equip \(\mathcal H_K\) with the real Euclidean inner product and its induced norm,
\begin{equation}
\label{eq:real-Euclidean-inner-product}
(z,w)_{\mathcal H_K}:=\operatorname{Re}(z^Hw)
\quad\text{and}\quad
\|z\|^2:=(z,z)_{\mathcal H_K}=z^Hz,
\end{equation}
respectively. All Jacobians w.r.t. a variable in \(z \in \mH_K\) are computed w.r.t. the real coordinates in \eqref{eq-2K+1}.
\end{remark}

Besides the one-sided state vector $\bar a(t) \in \mH_K$ in \eqref{eq-def-bar-a}, let us also define the diagonal matrix
\begin{equation}
\label{eq:general-Lambda}
\Lambda(P,\omega_0):=\operatorname{diag}\bigl(-P(i\omega_0 k)\bigr)_{k=0}^{K},
\end{equation}
and the $(K+1)$-dimensional vector $\eta(\bar a)$
\begin{equation}
\label{eq:eta-general}
[\eta(\bar a)]_k:=k\sum_{\ell=-K}^{K} \bar a_\ell \bar a_{k-\ell},
\qquad k=0,\ldots,K,
\end{equation}
where \(\bar a_{-\ell}= \bar a_\ell^*\) and \(\bar a_j=0\) for \(|j|>K\). The map $\eta$ defined in \eqref{eq:eta-general} is quadratic, from $\mathcal H_K$ into itself. Hence, the truncated master system \eqref{eq:master-truncated-ode} can be characterized on $\mathcal H_K$ by the ODE
\begin{equation}
\label{eq:master-vector-general}
\dot{\bar a}=\Lambda(P,\omega_0)\bar a-\frac{ic_1\omega_0}{2}\eta(\bar a).
\end{equation}

The slave model is designed with the same truncation and is driven by reconstructed master coefficients \(\hat a_k(t)\), \(k=0,\ldots,K\). These coefficients can be obtained from noiseless physical-space data by a finite Fourier reconstruction. Indeed, suppose that, at every time \(t\), the values of the truncated master field are available at sensor locations $\mathcal{X}_N=\{x_0,\ldots,x_{N-1}\} \subset [0,X)$, namely, one observes \(y_n(t)=u_K(t,x_n)\). Let \(\Psi_K\in\mathbb C^{N\times(2K+1)}\) be the matrix with entries \( [\Psi_K]_{n,k}=\phi_k(x_n)\), $|k| \le K$. If \(N\ge 2K+1\) and \(\Psi_K\) has full column rank, then the retained coefficients are recovered by the least-squares formula
\[
\hat a_{-K:K}(t)=(\Psi_K^H\Psi_K)^{-1}\Psi_K^H y(t),
\]
where $y(t)=( y_0(t), \ldots, y_{N-1}(t) )^\top$ and \(\hat a_{-K:K}(t)=(\hat a_{-K}(t),\ldots,\hat a_K(t))^\top\). 
Since \(u_K\) is real-valued, \(\hat a_{-k}(t)=\hat a_k^*(t)\), and it is sufficient to study the one-sided vector $\hat a(t) = (\hat a_0,\ldots,\hat a_K)^\top \in \mH_K$. For the truncated master field and noiseless reconstruction setting described above, $\hat a(t) = \bar a(t)$. If the observations are collected from the non-truncated field $u(t,x)$, then $\hat a(t) \approx \bar a(t)$ only.  

Similar to \cite{Miguez2024}, let us assume a truncated slave field of the form
\begin{equation}
\label{eq:slave-field-general}
v_K(t,x)=\sum_{k=-K}^{K}b_k(t)\phi_k(x),
\end{equation}
with \(b_{-k}=b_k^*\), and its retained coefficients satisfying
\begin{equation}
\label{eq:slave-general}
\dot b_k
=
-P(i\omega_0 k)b_k
-\frac{ic_1\omega_0 k}{2}\sum_{\ell=-K}^{K}b_\ell b_{k-\ell}
+\sum_{j=0}^K D_{k,j}(\hat a_j-b_j),
\qquad k=0,\ldots,K,
\end{equation}
where $D_{k,j}$ is the $k$-th row, $j$-th column entry of the $(K+1) \times (K+1)$ coupling matrix $D$. We assume that $Dz\in\mH_K$ for every $z \in \mathcal H_K$, so that the coupling preserves real-valued fields. This condition holds, in particular, for the scalar coupling \(D=dI\) used below. In vector form,
\beq
\label{eq:slave-vector-general}
\dot b=\Lambda(P,\omega_0)b-\frac{ic_1\omega_0}{2}\eta(b)+D(\hat a-b),
\quad \text{where} \quad b(t)=\left( b_0(t),\ldots,b_K(t) \right)^\top\in\mathcal H_K.
\eeq

%
\subsection{Local synchronization}

The synchronization-error field is
\begin{equation}
\label{eq:error-field-general}
\mE_K(t,x):=u_K(t,x)-v_K(t,x)=\sum_{k=-K}^{K}e_k(t)\phi_k(x),
\end{equation}
where $e_k(t)=\bar a_k(t)-b_k(t)$, and the corresponding stored error vector is \(e(t)=\bar a(t)-b(t)\in\mathcal H_K\). We aim to analyze the stability of the synchronization error signal $e(t)$ around $e(t)=0$. For this purpose, we impose some (fairly standard \cite{Miguez2024}) regularity assumptions.

\begin{assumption}[Exact reconstruction]
\label{ass:exact-reconstruction}
The observed master field is exactly the truncated field \(u_K\), and the reconstruction is exact on the retained nonnegative modes $\hat a_k(t)=\bar a_k(t)$, for $k=0,\ldots, K$ and $t\in[0,T]$.
\end{assumption}

\begin{assumption}[Bounded master]
\label{ass:bounded-master}
The truncated master field is uniformly bounded on the time interval of interest, i.e.,
\[
|u_K|_{\infty,T}:=\sup_{x\in\mathbb R,\,0\le t\le T}|u_K(t,x)|<\infty .
\]
\end{assumption}

\begin{assumption}[Scalar coupling]
\label{ass:scalar-coupling}
{
The synchronization gain is \(D=dI\), where \(I\) is the identity matrix and \(d>0\) is a real constant.
}
\end{assumption}

Under Assumption~\ref{ass:exact-reconstruction}, subtracting \eqref{eq:slave-vector-general} from \eqref{eq:master-vector-general} yields
\begin{equation}
\label{eq:error-ode-general}
\dot e={\sf f}(t,e),
\quad \text{where} \quad
{\sf f}(t,e):=\bigl(\Lambda(P,\omega_0)-dI\bigr)e
-\frac{ic_1\omega_0}{2}\left[\eta(\bar a(t))-\eta(\bar a(t)-e)\right].
\end{equation}
The zero field \(\mE_K(t,x)\equiv0\) in physical space is equivalent to the zero vector solution \(e(t)\equiv0\) in Fourier space. 


Since \(\eta\) is a homogeneous quadratic map on \(\mathcal H_K\), it admits a symmetric bilinear polarization \(B\), defined by \(B(z,h):=\frac14\bigl[\eta(z+h)-\eta(z-h)\bigr]\), such that \(\eta(w)=B(w,w)\) for every \(w\in\mathcal H_K\). For \(z,h\in\mathcal H_K\), we define
\beq
Q(z)h:=2B(z,h)=\eta(z+h)-\eta(z)-\eta(h).
\label{eq-Qzh}
\eeq
Expanding the polarization shows that, for fixed \(z\), the map \(h\mapsto Q(z)h\) is linear in the real coordinates of \(h\). Moreover, \(Q(z)h\) is precisely the directional (Fr\'echet) derivative of \(\eta\) at \(z\) in the direction \(h\),
\beq
D\eta(z)[h]=\lim_{\varepsilon\to0}\frac{\eta(z+\varepsilon h)-\eta(z)}{\varepsilon}=Q(z)h.
\label{eq-Frechet-der}
\eeq
Hence, when \(Q(z)\) is viewed as a matrix in the \(2K+1\) real coordinates of \(\mathcal H_K\) introduced in \eqref{eq-2K+1}, it is exactly the Jacobian matrix of \(\eta\) at \(z\). Note, however, that \(Q(z)\) is in general not complex-linear in \(h\). Since the negative modes are recovered by conjugation, its action in the one-sided complex coefficient notation
may involve both $h_k$ and $h_k^*$. For this reason, all Jacobians hereafter are understood in the real coordinates displayed in \eqref{eq-2K+1}.

Expanding the quadratic identity \(\eta(z-h)=\eta(z)-Q(z)h+\eta(h)\) gives, equivalently,
\[
\eta(z)-\eta(z-h)=Q(z)h-\eta(h),
\qquad z,h\in\mathcal H_K,
\]
an exact algebraic relation with no remainder. Combined with the scalar coupling of Assumption~\ref{ass:scalar-coupling}, it yields the exact decomposition of the drift in \eqref{eq:error-ode-general}. In particular, we obtain the identity
\begin{equation}
\label{eq:error-drift-decomposition}
{\sf f}(t,e)=A(t)e+r(e),
\end{equation}
where 
\begin{equation}
\label{eq:linearization-general}
A(t)=\Lambda(P,\omega_0)-dI-\frac{ic_1\omega_0}{2}Q(\bar a(t))
\end{equation}
is an identity of real-linear operators on $\mathcal H_K$, written in one-sided complex notation, and the real-coordinate matrix of $A(t)$ is the Jacobian of ${\sf f}(t,\cdot)$ at $e=0$. The remainder in \eqref{eq:error-drift-decomposition} is
\beq
r(e):=\frac{ic_1\omega_0}{2}\eta(e).
\eeq
Since \(\eta\) is quadratic on the finite-dimensional space \(\mathcal H_K\), there is a constant \(C_r<\infty\) such that
\begin{equation}
\label{eq:remainder-bound}
\|r(e)\|\le C_r\|e\|^2,
\quad \forall e\in\mathcal H_K.
\end{equation}

\begin{proposition}
\label{prop:local-sync-general}
Under Assumptions~\ref{ass:exact-reconstruction}--\ref{ass:scalar-coupling}, there exists \(d_0<\infty\) such that, for every \(d>d_0\), the zero-error solution \(e(t) = 0\) of the synchronization-error dynamics on \(\mathcal H_K\) is locally exponentially stable on \([0,T]\). 
\end{proposition}

See Appendix \ref{app:proof-prop1} for a proof. By Parseval's identity, $
\|\mE_K(t,\cdot)\|_{L^2(0,X)}^2
\le 2X\|e(t)\|^2
$, hence local exponential stability of $e\equiv 0$ implies local exponential stability of $\mE_K(t,\cdot)\equiv 0$ in $L^2(0,X)$.

\begin{remark}
\label{rem:non-diagonal-D}
The diagonal scalar coupling of Assumption \ref{ass:scalar-coupling} is not essential. The same local argument applies to any coupling matrix \(D\) that preserves \(\mathcal H_K\) and provides a positive spectral gap. 
\end{remark}

\begin{remark}
If Assumptions~\ref{ass:exact-reconstruction} and \ref{ass:bounded-master} are strengthened so that \(\hat a_k(t)=\bar a_k(t)\) for \(k=0,\ldots,K\) and all \(t\ge0\), and
\[
|u_K|_{\infty}:=\sup_{x\in\mathbb R,\,t\ge0}|u_K(t,x)|<\infty,
\]
then the coefficients \(\bar a_k(t)\) are uniformly bounded for all \(t\ge0\). In that case the same argument in the proof of Proposition \ref{prop:local-sync-general} yields constants \(\alpha>0\) and \(\rho>0\), independent of time, such that whenever \(\|e(0)\|<\rho\) there exists a constant $C<\infty$ satisfying 
\[
\|e(t)\|\le C {\rm e}^{-\alpha t}\|e(0)\|,
\quad \text{for all $t\ge0$}.
\]
Equivalently, $e(t) \to 0$ (and $\mE(t,x) \to 0$) exponentially fast.
\end{remark}

%
\section{Observational noise and stochastic slave model} \label{sStochasticSynch}

In practical applications, the observations from the master field can be expected to be contaminated by noise. In this section, we assume white Gaussian observational noise, which naturally yields a stochastic slave model. To be specific, the master system remains deterministic and belongs to the class \eqref{eq:intro-pde}, while the Fourier modes of the slave model become stochastic and can be described by an It\^o SDE.

%
\subsection{Noisy observations on a spatial grid} \label{ssNoisyObs}

Recall the spatial observation grid \(\mathcal{X}_N=\{x_0,\dots,x_{N-1}\}\subset[0,X)\). We assume that the deterministic master field \(u_K(t,x)\) is observed continuously in time, but only at the sites in \(\mathcal{X}_N\). The noiseless observation vector is
\[
u_N(t):=
\begin{bmatrix}
u_K(t,x_0)\\
\vdots\\
u_K(t,x_{N-1})
\end{bmatrix}\in\mathbb{R}^N
\]
and we model the instantaneous noisy observations formally as
\begin{equation}
\label{eq:noisy-observations}
z(t,x_i)=u_K(t,x_i)+\xi(t,x_i),
\qquad i=0,\dots,N-1,
\end{equation}
where \(\xi(t,x_i)\) denotes zero-mean Gaussian white noise. Let \(Z_N(t)\) denote the accumulated observation process. Its formal time derivative is the instantaneous observation vector $\dot Z_N(t)=\left[z(t,x_0),\ldots,z(t,x_{N-1})\right]^\top$, hence
\begin{equation}
\label{eq:noisy-observation-vector}
\dot Z_N(t)=u_N(t)+\xi_N(t),
\end{equation}
where \(\xi_N(t)=[\xi(t,x_0),\ldots,\xi(t,x_{N-1})]^\top\).
We assume that the observation noise is temporally white and independent across sensor locations, so it can be written as $\xi_N(t)=\sigma\dot W_N(t)$, where \(W_N\) is an \(N\)-dimensional standard Wiener process and \(\sigma>0\) is the noise intensity. The instantaneous observation equation \eqref{eq:noisy-observations} is understood in this formal white-noise sense. Its rigorous representation uses the accumulated observation process \(Z_N(t)\), initialized at \(Z_N(0)=0\), which satisfies the It\^o SDE
\begin{equation}
\label{eq:obs-sde}
{\rm d}Z_N(t)=u_N(t)\,{\rm d}t+\sigma\, {\rm d}W_N(t).
\end{equation}
The same differential observation convention is used by Br\"ocker et al.~\cite[Eq.~(1.3)]{BroeckerDelSartoHieberPalmaZoechling2026}.

%

\subsection{Least-squares reconstruction of the Fourier modes}

We now apply the spatial reconstruction map to the accumulated observations. Recall the Fourier matrix associated with the grid \(\mathcal X_N\), \(\Psi_K\in\mathbb C^{N\times(2K+1)}\), with columns ordered according to the retained indices $|k| \le K$, and let
\begin{equation}
\label{eq:R-matrix}
R:=(\Psi^H_K\Psi_K)^{-1}\Psi_K^H\in\mathbb{C}^{(2K+1)\times N}
\end{equation}
be the least-squares reconstruction matrix, assuming that \(N\ge 2K+1\) and \(\Psi_K^H\Psi_K\) is invertible. The matrix \(R\) reconstructs the full vector of retained coefficients from noiseless physical-space samples. Applying it to \(Z_N(t)\) instead gives the accumulated Fourier observations, denoted by \(\hat a_{-K:K}(t):=RZ_N(t)\). Since \(Z_N(t)\) is real-valued, the reconstructed modes are Hermitian and it is sufficient to work with the one-sided accumulated observation vector
\[
\hat a(t)=\left[\hat a_0(t),\ldots,\hat a_K(t)\right]^\top\in\mathcal H_K.
\]

Let \(R_K\in\mathbb C^{(K+1)\times N}\) be the submatrix of \(R\) formed by the rows corresponding to modes \(k=0,\ldots,K\). Conjugate symmetry of the Fourier reconstruction gives $R_K y \in \mH_K$ $\forall y \in \mbR^N$. Applying \(R_K\) to the observation SDE \eqref{eq:obs-sde} therefore yields an \(\mathcal H_K\)-valued process,
\begin{equation}
\label{eq:noisy-modes-sde}
{\rm d}\hat a(t)=R_K\,{\rm d}Z_N(t)=R_K u_N(t)\,{\rm d}t+\sigma R_K\,{\rm d}W_N(t),
\end{equation}
where \(\hat a(t):=R_KZ_N(t)\) contains the retained Fourier coordinates of the accumulated observation process, with \(\hat a(0)=0\). Assuming that the sampling grid has enough resolution for the least-squares reconstruction to be exact on the retained subspace, one has
$
R_K u_N(t)=\bar a(t)
$, where \(\bar a(t)\) is the deterministic master vector in \eqref{eq:master-vector-general}. Hence \eqref{eq:noisy-modes-sde} becomes
\begin{equation}
\label{eq:noisy-mode-estimate-ideal}
{\rm d}\hat a(t)=\bar a(t) {\rm d}t+\sigma R_K {\rm d}W_N(t).
\end{equation}
In particular, when \(\sigma=0\), \(\hat a(t)=\int_0^t\bar a(s)\,{\rm d}s\). Thus \(\hat a(t)\) is an accumulated coefficient observation in this framework, not an instantaneous estimate of \(\bar a(t)\).

%
\subsection{Stochastic slave equation and synchronization error process}

With the observation convention above, the coupling correction in the slave equation is expressed through increments of the reconstructed observation process,
\[
D\bigl({\rm d}\hat a(t)-b(t)\,{\rm d}t\bigr)
=
D\bigl(R_K\,{\rm d}Z_N(t)-b(t)\,{\rm d}t\bigr).
\]
Replacing the feedback in \eqref{eq:slave-vector-general} by this correction and substituting \eqref{eq:noisy-mode-estimate-ideal} yields the It\^o SDE
\begin{equation}
\label{eq:stochastic-slave}
{\rm d}b(t)=
\left[
\Lambda(P,\omega_0)b(t)
-\frac{i c_1\omega_0}{2}\eta(b(t))
+
D\bigl(\bar a(t)-b(t)\bigr)
\right]{\rm d}t
+
\sigma D R_K\,{\rm d}W_N(t),
\end{equation}
where the coupling matrix \(D\) is chosen to preserve \(\mathcal H_K\).



Recall the synchronization error \(e(t)=\bar a(t)-b(t)\). Since the master system is deterministic and satisfies \eqref{eq:master-vector-general}, subtracting \eqref{eq:stochastic-slave} from \eqref{eq:master-vector-general} yields the error SDE
\begin{equation}
\label{eq:stochastic-error}
{\rm d}e(t)=
\left[
\bigl(\Lambda(P,\omega_0)-D\bigr)e(t)
-\frac{i c_1\omega_0}{2}
\Bigl(\eta(\bar a(t))-\eta(\bar a(t)-e(t))\Bigr)
\right]{\rm d}t
-
\sigma D R_K\, {\rm d}W_N(t).
\end{equation}
Equation \eqref{eq:stochastic-error} is the stochastic counterpart of the deterministic error equation \eqref{eq:error-ode-general}. It shows explicitly how noisy observations of the deterministic master system induce a diffusion term in the slave dynamics and, consequently, in the synchronization error process.

The stochastic synchronization problem can be formulated as a stability problem for the SDE \eqref{eq:stochastic-error}. Since the additive noise continuously excites the slave, one cannot expect exact pathwise convergence \(e(t)\to 0\). Instead, the natural notions of synchronization are probabilistic, e.g., convergence in mean square to a noise-dependent neighbourhood of the origin, or bounds on tail probabilities, as explored below. 

\begin{remark}
In Section \ref{sMS} the one-sided vectors $\hat a(t)$, $b(t)$ and $e(t)$ were introduced as deterministic quantities representing the reconstructed Fourier modes of the truncated master system, the modes of the slave model, and the modes of the deterministic synchronization error. The  introduction of observational noise in Section \ref{ssNoisyObs} fundamentally changes the framework. While we keep the same notation for conciseness, hereafter $\hat a(t)$, $b(t)$ and $e(t)$ are It\^o processes representing the one-sided Fourier modes of the accumulated noisy observations, the stochastic slave model and the stochastic synchronization error, respectively.   
\end{remark}

%
\section{Stochastic stability of master-slave synchronization}
\label{sec:stochastic-synchronization}

We analyze the stochastic synchronization error in \eqref{eq:stochastic-error} by comparing it with the deterministic reference error that would be obtained from noiseless observations. For fixed \(K\) and bounded master and reference paths, a sufficiently strong scalar coupling yields mean-square error bounds on any prescribed finite horizon. We then give an explicit time-uniform dissipativity criterion which yields a mean-square bound independent of the time horizon. Finally, Section~\ref{sec:physical-error} combines the finite-horizon estimate with Fourier projection and Galerkin-error bounds to control the error between the full master field and the truncated slave. We also identify the additional assumptions needed for a subsequent long-time limit.

%
\subsection{Setup and assumptions}

Using the drift notation introduced in \eqref{eq:error-ode-general}, the SDE \eqref{eq:stochastic-error} can be written compactly as
\begin{equation}
\label{eq:stochastic-error-b}
{\rm d}e(t)={\sf f}(t,e(t))\,{\rm d}t-\sigma D R_K\,{\rm d}W_N(t).
\end{equation}
Throughout this section, we measure Fourier coefficients in the norm induced by the physical-space \(L^2(0,X)\) norm. To be specific, let \(z\in\mathcal H_K\) and extend the coefficients by conjugate symmetry, then write
\[
\mathcal F^{-1}z(x)
:=z_0+\sum_{k=1}^K\left(z_k{\rm e}^{ik\omega_0x}+z_k^*{\rm e}^{-ik\omega_0x}\right).
\]
The corresponding real inner product and norm are, respectively,
\begin{equation}
\label{eq:fourier-L2-inner-product}
\langle z,w\rangle_{K;L^2}
:=X\left(z_0w_0+2\operatorname{Re}\sum_{k=1}^Kz_k^*w_k\right)
\quad \text{and} \quad
\|z\|_{K;L^2}^2:=\langle z,z\rangle_{K;L^2}.
\end{equation}
Parseval's identity then gives
\begin{equation}
\label{eq:parseval-induced-norm}
\|z\|_{K;L^2}^2=\|\mathcal F^{-1}z\|_{L^2(0,X)}^2,
\end{equation}
where \(\|v\|_{L^2(0,X)}^2=\int_0^X|v(x)|^2\,{\rm d}x\). This weighted norm is equivalent to the coefficient norm $\| z \|^2 = z^H z$ used in Section~\ref{sMS}, since
$
X\|z\|^2\le\|z\|_{K;L^2}^2\le2X\|z\|^2.
$

For a $(K+1) \times N$ matrix $A$, we use the weighted Frobenius (Hilbert--Schmidt) norm
\begin{equation}
\label{eq:induced-HS-norm}
\|A\|_{K;L^2,F}^2
:=\sum_{j=1}^N\|A\beta_j\|_{K;L^2}^2
=X\sum_{j=1}^N|A_{0j}|^2+2X\sum_{k=1}^K\sum_{j=1}^N|A_{kj}|^2,
\end{equation}
where \(\{\beta_j\}_{j=1}^N\) is the canonical basis of \(\mathbb R^N\). The zero mode is retained in both norms because the projected noise may perturb it.

We adopt Assumptions~\ref{ass:bounded-master} and \ref{ass:scalar-coupling}, as in Section \ref{sMS}, together with the exact-reconstruction condition \(R_Ku_N(t)=\bar a(t)\) for all \(t\in[0,T]\) and note that Assumption~\ref{ass:bounded-master} implies the existence of \(C_{a,T}<\infty\) such that
\begin{equation}
\label{eq:master-coeff-bound-V}
\sup_{0\le t\le T}\|\bar a(t)\|_{K;L^2}\le C_{a,T}.
\end{equation}
We also introduce a deterministic reference path used for comparison with the stochastic error.

\begin{assumption}[Reference error path]
\label{ass:reference-error-path}
The deterministic reference error \(e_{\mathrm{ref}}(t)\in\mathcal H_K\), \(t\in[0,T]\), solves the ODE
\begin{equation}
\label{eq:det-ref-error}
\dot e_{\mathrm{ref}}(t)={\sf f}\bigl(t,e_{\mathrm{ref}}(t)\bigr),
\end{equation}
with the same initial slave state as the stochastic model, i.e.,
\begin{equation}
\label{eq:common-initialization}
e(0)=e_{\mathrm{ref}}(0).
\end{equation}
Moreover, there exist a coupling gain $d_{\rm ref}<\infty$ and a constant $C_{{\rm ref},T}<\infty$, possibly depending on $K$ and $T$, such that for every coupling gain $d \ge d_{\rm ref}$, the reference path is bounded on \([0,T]\), namely, 
\begin{equation}
\label{eq:reference-bound}
\sup_{0\le t\le T}\|e_{\mathrm{ref}}(t)\|_{K;L^2}\le C_{\mathrm{ref},T}.
\end{equation}
\end{assumption}

The bounded master and reference paths are intended to control the linearization of the quadratic drift along the reference path. We assume scalar coupling, $D=dI$, for simplicity.

%
\subsection{Control of the linearization along the reference trajectory}
\label{ssControl}

For each fixed $t\in[0,T]$, the drift ${\sf f}(t,\cdot)$ is a polynomial map of
degree two on the real vector space $\mathcal H_K$ and is therefore
Fr\'echet differentiable. Its derivative at $e$ in the direction
$h\in\mathcal H_K$ is given by the real-linear operator
\begin{equation}
\label{eq:jacobian-ref-general}
\partial_e{\sf f}(t,e)[h]
=\bigl(\Lambda(P,\omega_0)-dI\bigr)h
-\frac{ic_1\omega_0}{2}\,Q\bigl(\bar a(t)-e\bigr)h,
\end{equation}
where $Q(z)h=D\eta(z)[h]$; see Eq. \eqref{eq-Frechet-der}. Note that \eqref{eq:jacobian-ref-general} is an identity between real-linear operators on $\mathcal H_K$, not between complex matrices: both $\Lambda(P,\omega_0)$ and $Q(z)$ preserve the Hermitian symmetry of their argument, so they map $\mathcal H_K$ into itself and act real-linearly on it, but applying either matrix to a vector that is not Hermitian-symmetric has no meaning for the derivative. When the operator $h\mapsto\partial_e{\sf f}(t,e)[h]$ is represented in the $2K+1$ real coordinates of \eqref{eq-2K+1}, it becomes an ordinary $(2K+1)\times(2K+1)$ real Jacobian matrix. With some abuse of notation, the same symbol $\partial_e{\sf f}(t,e)$, and the abbreviation $J_{\mathrm{ref}}(t):=\partial_e{\sf f}(t,e_{\mathrm{ref}}(t))$, are used below for both the operator on $\mathcal H_K$ and its real-coordinate matrix representation. The following lemma provides a control on the quadratic form of $J_{\mathrm{ref}}(t)$ in the weighted inner product \eqref{eq:fourier-L2-inner-product}.

\begin{lemma}
\label{lem:ref-dissipativity}
If \(R_Ku_N(t)=\bar a(t)\) for all \(t\in[0,T]\) and Assumptions~\ref{ass:bounded-master}, \ref{ass:scalar-coupling} and \ref{ass:reference-error-path} hold, then there exists \(d_0<\infty\) such that, for every \(d>d_0\), the linearization \(J_{\mathrm{ref}}(t):=\partial_e{\sf f}(t,e_{\mathrm{ref}}(t))\) along the reference path, represented in the \(2K+1\) real coordinates of \eqref{eq-2K+1}, satisfies
\begin{equation}
\label{eq:ref-dissipativity}
\langle z,J_{\mathrm{ref}}(t)z\rangle_{K;L^2}
\le-\alpha\|z\|_{K;L^2}^2,
\quad \text{for all}~~ z\in\mathcal H_K ~~\text{and}~~ t\in[0,T],
\end{equation}
for some \(\alpha>0\). The threshold \(d_0\) may depend on \(K\) and \(T\).
\end{lemma}

See Appendix~\ref{ap:ref-dissipativity} for a proof.

%
\subsection{Finite-horizon stochastic closeness}
\label{sec:finite-horizon-closeness}

Define the \(\mathcal H_K\)-valued deviation process \(\Delta(t):=e(t)-e_{\mathrm{ref}}(t)\), and let \(h_\Delta(t,\cdot):=\mathcal F^{-1}\Delta(t)\) be its physical-space field. We use $\Delta(t)$ to quantify the discrepancy between the stochastic synchronization error $e(t)$ and its deterministic counterpart $e_{\rm ref}(t)$ which satisfies the same ODE as in Proposition \ref{prop:local-sync-general}. Subtracting the reference ODE \eqref{eq:det-ref-error} from \eqref{eq:stochastic-error-b} yields
\begin{equation}
\label{eq:Delta-SDE}
{\rm d}\Delta(t)
=\Bigl({\sf f}\bigl(t,e_{\mathrm{ref}}(t)+\Delta(t)\bigr)
-{\sf f}\bigl(t,e_{\mathrm{ref}}(t)\bigr)\Bigr)\,{\rm d}t
-\sigma d R_K\,{\rm d}W_N(t),
\quad \text{with} \quad \Delta(0)=0,
\end{equation}
and the exact quadratic identity \eqref{eq:error-drift-decomposition} leads to
\begin{equation}
\label{eq:Delta-linearization}
{\sf f}(t,e_{\mathrm{ref}}+\Delta)-{\sf f}(t,e_{\mathrm{ref}})
=J_{\mathrm{ref}}(t)\Delta+r(\Delta),
\quad \text{where} \quad 
r(\Delta)=\frac{ic_1\omega_0}{2}\eta(\Delta).
\end{equation}

The quadratic remainder $r(\Delta)$ has zero contribution to the energy. To see this, let
\[
V_K:=\operatorname{span}\{{\rm e}^{ik\omega_0x}:|k|\le K\}
\]
be the order-\(K\) Fourier--Galerkin space, and recall the orthogonal projection \({\rm P}_K:L^2(0,X)\to V_K\). For \(z\in\mathcal H_K\), the real field \(h_z:=\mathcal F^{-1}z\) belongs to \(V_K\), and the definition of \(\eta\) in \eqref{eq:eta-general} gives \(\mathcal F^{-1}r(z)=c_1{\rm P}_K(h_z\partial_xh_z)\) via differentiation and convolution. Moreover, for any pair $z,w \in \mH_K$, it is not difficult to show that
\beq
\langle z, w \rangle_{K;L^2} = \int_0^X h_z h_w {\rm d}x, 
\quad \text{where} \quad
h_z = \mF^{-1} z ~~\text{and}~~ h_w=\mF^{-1} w. 
\label{eq-parseval-inner-prod}
\eeq
The identity \eqref{eq-parseval-inner-prod}, combined with orthogonality and periodicity, readily implies
\begin{equation}
\label{eq:nonlinear-energy-cancellation}
\langle z,r(z)\rangle_{K;L^2}
=c_1\int_0^X h_z^2\partial_xh_z\,{\rm d}x
=\frac{c_1}{3}\int_0^X\partial_x(h_z^3)\,{\rm d}x=0.
\end{equation}
As a consequence, the quadratic-form bound in Lemma~\ref{lem:ref-dissipativity} controls the full drift difference in \eqref{eq:Delta-SDE} for every \(\Delta\in\mathcal H_K\), without any restriction on $\| \Delta \|_{K;L^2}$.

\begin{proposition}[Finite-horizon stochastic closeness]
\label{prop:stochastic-closeness}
Assume \(R_Ku_N(t)=\bar a(t)\) for all \(t\in[0,T]\), let Assumptions~\ref{ass:bounded-master}, \ref{ass:scalar-coupling} and \ref{ass:reference-error-path} hold on \([0,T]\), and choose \(d>d_0\) so that \eqref{eq:ref-dissipativity} is satisfied with \(\alpha>0\), as guaranteed by Lemma~\ref{lem:ref-dissipativity}. Then \eqref{eq:Delta-SDE} has a unique non-explosive strong solution in \(\mathcal H_K\) on \([0,T]\). For every \(t\in[0,T]\),
\begin{equation}
\label{eq-xx1}
\mbE\|\Delta(t)\|_{K;L^2}^2
\le\frac{\sigma^2d^2\|R_K\|_{K;L^2,F}^2}{2\alpha}
\left(1-{\rm e}^{-2\alpha t}\right).
\end{equation}
Moreover, there is a finite constant \(C_T\), depending only on \(T\), such that
\begin{equation}
\label{eq:stochastic-closeness-bound}
\mbE\left[\sup_{0\le s\le T}\|\Delta(s)\|_{K;L^2}^2\right]
\le C_T\,\sigma^2d^2\|R_K\|_{K;L^2,F}^2,
\end{equation}
and, for every \(\ell>0\), it follows that
\begin{equation}
\label{eq:Delta-prob-bound}
\mathbb P\left(\sup_{0\le s\le T}\|\Delta(s)\|_{K;L^2}>\ell\right)
\le\frac{C_T\,\sigma^2d^2\|R_K\|_{K;L^2,F}^2}{\ell^2}.
\end{equation}
\end{proposition}

See Appendix \ref{ap:prop:stochastic-closeness} for a proof.

The energy cancellation in \eqref{eq:nonlinear-energy-cancellation} makes \eqref{eq:ref-dissipativity} equivalent to the same quadratic bound for the full drift difference in \eqref{eq:Delta-linearization}, for every \(\Delta\in\mathcal H_K\). However, the gain threshold and the decay rate may depend on \(T\) and \(K\). A time-uniform conclusion requires a single gain and a positive decay rate valid on every finite horizon. This is investigated in the next subsection.

%
\subsection{Time-uniform stochastic closeness}
\label{sec:time-uniform-closeness}

We now give conditions ensuring that the decay rate in Proposition~\ref{prop:stochastic-closeness} is independent of the time horizon. We first state the dissipativity criterion on \([0,T]\), so it can also be used for the finite-horizon estimates in Section~\ref{sec:physical-error}.

Consider the real part of the polynomial $-P(i\xi)$,
\begin{equation}
\label{eq:qstar-def}
\mathfrak q(\xi):=\operatorname{Re}[-P(i\xi)]
=c_2\xi^2-c_4\xi^4+c_6\xi^6,
\quad \xi\in\mathbb R,
\end{equation}
and assume
\begin{equation}
\label{eq:qstar-bounded}
\mathfrak q_*:=\sup_{\xi\in\mathbb R}\mathfrak q(\xi)<\infty.
\end{equation}
This holds, for instance, for the generalized Kuramoto--Sivashinsky equation with \(c_4>0\), Burgers' equation with \(c_2\le0\), and Kawahara-type equations with \(\mathfrak q\equiv0\).

For a given truncation order $K$, let \(\mE_{\mathrm{ref},K}:=\mathcal F^{-1}e_{\mathrm{ref}}\) and define the reference slave field
\[
v_{\mathrm{ref},K}(t,x):=u_K(t,x)-\mE_{\mathrm{ref},K}(t,x),
\]
then assume
\begin{equation}
\label{eq:ref-slave-Linfty-bound}
B_{\mathrm{ref},K}:=\sup_{0\le t\le T}\|v_{\mathrm{ref},K}(t,\cdot)\|_{L^\infty(0,X)}<\infty.
\end{equation}
For any fixed \(K\), Bernstein's inequality yields
\begin{equation}
\label{eq:ref-slave-derivative-bound}
\|\partial_xv_{\mathrm{ref},K}(t,\cdot)\|_{L^\infty(0,X)}
\le K\omega_0\|v_{\mathrm{ref},K}(t,\cdot)\|_{L^\infty(0,X)},
\end{equation}
and therefore
\begin{equation}
\label{eq:L-ref-def}
L_{\mathrm{ref},K}:=\sup_{0\le t\le T}\|\partial_xv_{\mathrm{ref},K}(t,\cdot)\|_{L^\infty(0,X)}
\le K\omega_0B_{\mathrm{ref},K}<\infty.
\end{equation}
This implication uses the finite Fourier truncation and does {\em not} give a derivative bound uniform in \(K\).

\begin{lemma}
\label{lem:global-dissipativity-fourier}
Assume \eqref{eq:qstar-bounded} and \eqref{eq:ref-slave-Linfty-bound} hold, let \(D=dI\), and choose
\begin{equation}
\label{eq:global-dissipativity-coupling}
d>\mathfrak q_*+\frac{|c_1|}{2} L_{\mathrm{ref},K}.
\end{equation}
Then, with \(\mu_K := d-\mathfrak q_*-\frac{|c_1|}{2} L_{\mathrm{ref},K}>0\), the inequality
\begin{equation}
\label{eq:global-dissipativity}
\bigl\langle\Delta,
{\sf f}(t,e_{\mathrm{ref}}(t)+\Delta)-{\sf f}(t,e_{\mathrm{ref}}(t))
\bigr\rangle_{K;L^2}
\le -\mu_K\|\Delta\|_{K;L^2}^2
\end{equation}
holds for every \(t\in[0,T]\) and \(\Delta\in\mathcal H_K\). In particular, \eqref{eq:ref-dissipativity} holds with \(\alpha=\mu_K\).
\end{lemma}

See Appendix \ref{ap:lem:global-dissipativity-fourier} for a proof. 

For fixed \(K\), the same proof works with the always finite quantity \(\max_{|k|\le K}\mathfrak q(\omega_0k)\) in place of \(\mathfrak q_*\). Condition~\eqref{eq:qstar-bounded} is useful because the spectral bound then does not depend on \(K\). To obtain a time-uniform estimate, we additionally require bounds on the reference field that do not depend on \(T\).

\begin{proposition}[Time-uniform stochastic closeness]
\label{prop:global-stochastic-closeness}
Let Assumptions~\ref{ass:bounded-master}, \ref{ass:scalar-coupling} and \ref{ass:reference-error-path} hold on every finite interval with a single coupling gain \(d>0\), and assume \(R_Ku_N(t)=\bar a(t)\) for all \(t\ge0\). Let \eqref{eq:qstar-bounded} hold and suppose the reference field satisfies
\[
B_{\mathrm{ref},K}^\infty:=\sup_{t\ge0}\|v_{\mathrm{ref},K}(t,\cdot)\|_{L^\infty(0,X)}<\infty
\quad \text{and} \quad
L_{\mathrm{ref},K}^\infty:=\sup_{t\ge0}\|\partial_xv_{\mathrm{ref},K}(t,\cdot)\|_{L^\infty(0,X)}<\infty.
\]
If $\mu_K:=d-\mathfrak q_*-\frac{|c_1|}{2}L_{\mathrm{ref},K}^\infty>0$, then \eqref{eq:Delta-SDE} has a unique global strong solution in \(\mathcal H_K\), and
\begin{equation}
\label{eq:global-limsup}
\sup_{t\ge0}\mbE\|\Delta(t)\|_{K;L^2}^2
\le\frac{\sigma^2d^2\|R_K\|_{K;L^2,F}^2}{2\mu_K}.
\end{equation}
\end{proposition}

\begin{proof}
Lemma~\ref{lem:global-dissipativity-fourier} gives \eqref{eq:ref-dissipativity} with \(\alpha=\mu_K\) on every finite interval. Proposition~\ref{prop:stochastic-closeness} therefore applies on each such interval. Uniqueness makes these solutions agree on overlapping intervals, yielding a global solution. Its pointwise estimate \eqref{eq-xx1}, with the same \(\mu_K>0\) for every horizon, gives \eqref{eq:global-limsup} upon taking the supremum over \(t\ge0\).
\end{proof}


%
\subsection{Error bounds in physical space}
\label{sec:physical-error}

We now transfer the weighted Fourier-space estimates to physical space and account for the approximation errors between the original (infinite dimensional) master field $u(t,x)$ and its Fourier--Galerkin truncation $u_K(t,x)$. The distinction between $u$ and $u_K$ in the analysis requires notation that makes the dependence on the truncation order $K$ explicit. For this reason, hereafter we write $e_K(t)$, $e_{{\rm ref},K}(t)$ and $\Delta_K(t)$ for the one-sided Fourier coefficient vectors that we have previously denoted as $e(t)$, $e_{\rm ref}(t)$ and $\Delta(t)$, respectively. 

Accordingly, we also distinguish two physical-space synchronization error fields. For a truncation order $K$ we recall the error field $\mE_K(t,x) = \mathcal F^{-1} e_K(t)(x)$ in \eqref{eq:error-field-general}, and also define 
\begin{equation}
\label{eq:physical-error-full}
\mE_\infty^{(K)}(t,x):=u(t,x)-v_K(t,x),
\end{equation}
which is the physical-space error between the infinite-dimensional master field $u$ and the truncated slave field $v_K$. The field $\mE_\infty^{(K)}$ depends on the truncation order through $v_K$. If we also recall the synchronization error field for the deterministic reference system, $\mE_{{\rm ref},K}:=\mathcal F^{-1} e_{{\rm ref},K}$, then the relationships $e_K=e_{{\rm ref},K}+\Delta_K$ and $h_{\Delta,K}=\mF^{-1}\Delta_K$, combined with Parseval's identity, yield
\begin{equation}
\label{eq:physical-error-truncated-decomposition}
\mE_K(t,x)=\mE_{{\rm ref},K}(t,x)+h_{\Delta,K}(t,x),
\quad\text{where}\quad
\|h_{\Delta,K}(t,\cdot)\|_{L^2(0,X)}^2
=\|\Delta_K(t)\|_{K;L^2}^2.
\end{equation}
Equation~\eqref{eq:physical-error-truncated-decomposition} shows that the weighted bounds in Proposition~\ref{prop:stochastic-closeness} can provide a control on the stochastic part of the physical synchronization error on \([0,T]\). Under its stronger hypotheses, Proposition~\ref{prop:global-stochastic-closeness} also provides time-uniform mean-square bounds.

The analysis of the infinite-dimensional error field $\mE_\infty^{(K)}$ requires a quantification of the spectral truncation error. Let ${\rm P}_K$ be the order-$K$ Fourier projection introduced in Section~\ref{sPDEFamily}, and define
\begin{equation}
\label{eq:spectral-error-components}
\zeta_K:=u - {\rm P}_Ku = (I-{\rm P}_K)u
\quad \text{and} \quad
\gamma_K:={\rm P}_Ku-u_K,
\end{equation}
where $I$ denotes the identity operator. The infinite-dimensional field $\zeta_K(t,x)$ is the unresolved tail after projection: it contains only Fourier modes with $|k|>K$ and measures the error of the best $L^2$-projection of $u(t,x)$ onto the retained Fourier subspace. The field $\gamma_K(t,x)$ is finite-dimensional and quantifies the difference between the projected full solution and the solution of the closed Fourier--Galerkin system. In general, $\gamma_K\ne0$ because ${\rm P}_K(uu_x)$ is not determined only by the retained coefficients of $u$. Combining \eqref{eq:physical-error-truncated-decomposition} and \eqref{eq:spectral-error-components} yields the four-term decomposition
\begin{equation}
\label{eq:physical-error-full-decomposition}
\mE_\infty^{(K)}(t,x)
=\zeta_K(t,x)+\gamma_K(t,x)+\mE_{{\rm ref},K}(t,x)+h_{\Delta,K}(t,x).
\end{equation}

For a test function $w(x) \in H_{\rm per}^s(0,X)$ with Fourier coefficients $\{\widehat w_k\}_{k\in\mathbb Z}$, recall the periodic Sobolev norm
\beq
\|w\|_{H^s_{\rm per}(0,X)}^2
:=X\sum_{k\in\mathbb Z}(1+\omega_0^2k^2)^s|\widehat w_k|^2
\nn
\eeq
and set $\kappa_K:=\sqrt{1+\omega_0^2(K+1)^2}$. Assuming $u(t,x) \in H_{\rm per}^s(0,X)$, an application of Parseval's identity immediately yields
\begin{equation}
\label{eq:projection-tail-Hs}
\|\zeta_K(t,\cdot)\|_{L^2(0,X)}
\le \kappa_K^{-s}\|u(t,\cdot)\|_{H^s_{\rm per}(0,X)}
~~ \text{and} ~~
\|\partial_x\zeta_K(t,\cdot)\|_{L^2(0,X)}
\le \kappa_K^{-(s-1)}\|u(t,\cdot)\|_{H^s_{\rm per}(0,X)}.
\end{equation}

We now collect the hypotheses needed to obtain bounds whose constants do not deteriorate as the truncation order $K$ increases.

\begin{assumption}[Uniform spectral approximation]
\label{ass:uniform-physical-error}
Fix $K_0\ge1$ and $s>\frac{3}{2}$. The following conditions hold.

\begin{enumerate}[label=\textup{(\alph*)},ref=\theassumption(\alph*),leftmargin=.25in]

\item \label{ass:uniform-physical-error-regularity}
The bound $\mathfrak{q}_*$ \eqref{eq:qstar-bounded} is satisfied, $u_K(0,\cdot)={\rm P}_Ku(0,\cdot)$ for every $K\ge K_0$, and there is a constant $M_{s,T}<\infty$, independent of $K$, such that
\begin{equation}
\label{eq:uniform-Sobolev-assumption}
\sup_{0\le t\le T}\|u(t,\cdot)\|_{H^s_{\rm per}(0,X)}
+\sup_{K\ge K_0}\sup_{0\le t\le T}
\left(\|u_K(t,\cdot)\|_{H^s_{\rm per}(0,X)}+\|v_{{\rm ref},K}(t,\cdot)\|_{H^s_{\rm per}(0,X)}\right)
\le M_{s,T}.
\end{equation}

\item \label{ass:uniform-physical-error-reference}
For every $K\ge K_0$, the deterministic reference error $e_{{\rm ref},K}$ solves the ODE \eqref{eq:det-ref-error} and the stochastic and deterministic slaves have the same initial state, so that $e_K(0)=e_{{\rm ref},K}(0)$. Moreover,
\begin{equation}
\label{eq:uniform-reference-initial-error}
E_0^2:=\sup_{K\ge K_0}\|e_{{\rm ref},K}(0)\|_{K;L^2}^2<\infty.
\end{equation}

\item \label{ass:uniform-physical-error-gain}
The coupling is scalar, $D=dI$, where the same $d>0$ is used for every $K\ge K_0$. Let $C_{\rm emb} = C_{\rm emb}(s,X)$ denote the norm of the embedding $H^s_{\rm per}(0,X)\hookrightarrow W^{1,\infty}(0,X)$ and let $\overline L:=C_{\rm emb}M_{s,T}$. We assume that $d$ is large enough to ensure
\begin{equation}
\label{eq:uniform-gain-K}
\mu_0:=d-\mathfrak q_*-\frac{|c_1|}{2}\overline L>0.
\end{equation}

\item \label{ass:uniform-physical-error-grid}
For every $K\ge K_0$, the observations are collected on an equispaced grid $\mathcal X_{N_K}$ with $N_K\ge2K+1$ points, and the retained modes are reconstructed with the least-squares matrix $R_K$ defined in \eqref{eq:R-matrix}.
\end{enumerate}
\end{assumption}

In Assumption~\ref{ass:uniform-physical-error-grid}, as throughout the stochastic observation model, the observations are samples of the truncated master field \(u_K\), rather than of the full field \(u\). Hence the condition \(N_K\ge 2K+1\) guarantees exact reconstruction of the retained noiseless modes. If the full field \(u\) were observed instead, unresolved modes could alias into the reconstruction, and the grid condition alone would not guarantee exact recovery.

The error field $\gamma_K(t,x)$ is controlled by the following consequence of Assumption~\ref{ass:uniform-physical-error-regularity}.

%
\begin{lemma}[Uniform Fourier--Galerkin error]
\label{lem:uniform-Galerkin-error}
If Assumption~\ref{ass:uniform-physical-error-regularity} holds, then there
exists a constant $G_{s,T}<\infty$, independent of $K$, such that, for every
$K\ge K_0$,
\begin{equation}
\label{eq:Galerkin-error-Hs}
\sup_{0\le t\le T}\|\gamma_K(t,\cdot)\|_{L^2(0,X)}^2
\le G_{s,T}\kappa_K^{-2(s-1)}.
\end{equation}
The constant $G_{s,T}$ may depend on $s$, $T$, $X$, the PDE coefficients,
and $M_{s,T}$, but not on the truncation order $K$. 
\end{lemma}

A proof is given in Appendix~\ref{ap:Galerkin-error-Hs}. The Fourier projection estimates used in the proof are standard; see Section 5.1.2 of Ref.~\cite{Canuto06}. The energy-error framework for linear evolution problems is developed in Section 6.5.1 of \cite{Canuto06}. Appendix~\ref{ap:Galerkin-error-Hs} provides an adaptation to the nonlinear PDE family considered here.

Assumption~\ref{ass:uniform-physical-error} supplies the exact-reconstruction condition \(R_Ku_N(t)=\bar a(t)\), with \(N=N_K\), and the earlier Assumptions~\ref{ass:bounded-master} and \ref{ass:scalar-coupling} on each finite horizon. The Sobolev bounds control \(u_K\) in \(L^\infty(0,X)\), \(v_{{\rm ref},K}\) in \(W^{1,\infty}(0,X)\), and \(e_{{\rm ref},K}\) in the weighted coefficient norm, uniformly in \(K\) and \(t\in[0,T]\). The reference dynamics and common initialization are specified in Assumption~\ref{ass:uniform-physical-error-reference}, the scalar gain in Assumption~\ref{ass:uniform-physical-error-gain}, and exact reconstruction in Assumption~\ref{ass:uniform-physical-error-grid}.

Lemma~\ref{lem:global-dissipativity-fourier} is granted by Assumptions~\ref{ass:uniform-physical-error-regularity} and \ref{ass:uniform-physical-error-gain} with a gap at least \(\mu_0\) for every \(K\ge K_0\). By the lemma's final assertion, \eqref{eq:ref-dissipativity} holds with \(\alpha=\mu_0\). Moreover, the orthogonality of the Fourier basis functions on the grids $\mathcal X_{N_K}$ in Assumption~\ref{ass:uniform-physical-error-grid} ensures that
\begin{equation}
\label{eq:uniform-reconstruction-bound}
\overline R^2:=\sup_{K\ge K_0}\|R_K\|_{K;L^2,F}^2
=\sup_{K\ge K_0}\frac{X(2K+1)}{N_K}\le X \quad \text{independently of $K$.}
\end{equation}

For the high-resolution ($K\to\infty$) and long-time ($t\to\infty$) inequality \eqref{eq:physical-error-time-K-limit} below, Assumption~\ref{ass:uniform-physical-error} must additionally hold on \([0,\infty)\), with \(M_{s,T}\) replaced by a constant \(M_s\) independent of the horizon. Then \(\overline L\) and \(\mu_0\) are independent of time. No time-uniform version of Lemma~\ref{lem:uniform-Galerkin-error} is required: for any fixed \(t<\infty\), its finite constant \(G_{s,T}\) on an interval containing \(t\) is independent of \(K\), so the Galerkin contribution vanishes when \(K\to\infty\), even if \(G_{s,T}\) grows with \(T\).

\begin{proposition}
\label{prop:physical-error-uniform-K}
If Assumption~\ref{ass:uniform-physical-error} holds, then, for every $K\ge K_0$ and $t\in[0,T]$,
\begin{equation}
\label{eq:physical-error-uniform-K}
\mbE\|\mE_\infty^{(K)}(t,\cdot)\|_{L^2(0,X)}^2
\le
4M_{s,T}^2\kappa_K^{-2s}
+4G_{s,T}\kappa_K^{-2(s-1)}
+4E_0^2{\rm e}^{-2\mu_0t}
+\frac{2\sigma^2d^2\overline R^2}{\mu_0}
\left(1-{\rm e}^{-2\mu_0t}\right).
\end{equation}
In particular,
\begin{equation}
\label{eq:physical-error-K-limit}
\limsup_{K\to\infty}
\mbE\|\mE_\infty^{(K)}(t,\cdot)\|_{L^2(0,X)}^2
\le
4E_0^2{\rm e}^{-2\mu_0t}
+\frac{2\sigma^2d^2\overline R^2}{\mu_0}
\left(1-{\rm e}^{-2\mu_0t}\right).
\end{equation}
If, in addition, Assumption~\ref{ass:uniform-physical-error} holds on $[0,\infty)$ with $M_{s,T}$ replaced by a constant $M_s$ independent of the time horizon $T$, then
\begin{equation}
\label{eq:physical-error-time-K-limit}
\limsup_{t\to\infty} \left(
	\limsup_{K\to\infty} \mbE\|\mE_\infty^{(K)}(t,\cdot)\|_{L^2(0,X)}^2
\right)
\le
\frac{2\sigma^2d^2\overline R^2}{\mu_0}.
\end{equation}
\end{proposition}

See Appendix \ref{ap:prop:physical-error-uniform-K} for a proof.


The smoothness hypothesis in Proposition~\ref{prop:physical-error-uniform-K} is Sobolev rather than analytic. The threshold $s>\frac{3}{2}$ is used to control the first spatial derivatives uniformly through the embedding $H^s_{\rm per}(0,X)\hookrightarrow W^{1,\infty}(0,X)$. The Fourier-tail estimate itself only requires $s>0$. For viscous Burgers and for members whose highest even-order term is dissipative, including the standard dissipative Kuramoto--Sivashinsky, Benney--Lin, and Nikolaevskiy equations, parabolic regularization and energy estimates make such bounds natural for smooth data and, often, yield stronger Gevrey regularity for positive times \cite{BaeBiswas2015}. For purely dispersive members, such as the Kawahara equation without dissipative terms, regularity is propagated rather than generated. The justification of the finite-horizon assumption requires sufficiently smooth initial data and well-posedness analysis \cite{BiagioniLinares1997}. 

The inequality \eqref{eq:physical-error-time-K-limit} is stronger: it requires time-uniform Sobolev bounds for the full and Galerkin master solutions and the reference slaves, as ensured, for example, by an absorbing set or suitable global a priori estimates. It does not, however, require the Galerkin constant $G_{s,T}$ in \eqref{eq:Galerkin-error-Hs} to be uniform in $T$, because the high-resolution limit is taken before the long-time limit. 

%
\section{Numerical results}
\label{sec:numerical-results}

We illustrate numerically the two main contributions to the synchronization error in the physical domain identified in Section~\ref{sec:physical-error}: the mean-square error induced by observation noise scales as $\mathcal O(\sigma^2)$ and, when the slave has a lower spectral resolution than the master, a nonzero error floor due to the spectral truncation of the slave. This error floor decreases as the slave truncation order is increased.

%
\subsection{Numerical setup}

For the computer experiments, we consider the canonical Kawahara equation
\begin{equation}
u_t+u u_x+u_{xxx}-u_{xxxxx}=0
\label{eq:numerical-kawahara}
\end{equation}
and the Nikolaevskiy equation
\begin{equation}
u_t+u u_x+q u_{xx}-u_{xxxx}-u_{xxxxxx}=0,
\qquad q=0.8.
\label{eq:numerical-nikolaevskiy}
\end{equation}
In the notation of Section~\ref{sPDEFamily}, their coefficient vectors are, respectively,
\[
(c_0, c_1, c_2, \ldots,c_6)=(0,1,0,1,0,-1,0) 
\quad \text{and} \quad 
(c_0, c_1, c_2, \ldots,c_6)=(0,1,0.8,0,-1,0,-1).
\]
The two models are simulated on the periodic interval $[0,X]$, with $X=120$.

The full master field is approximated by a Fourier--Galerkin model of order $K_*=64$, which is used as the high-resolution reference solution. For a slave of order $K \le K_*$, we therefore compute the numerical error field
\begin{equation}
\mE_{K_*}^{(K)}(t,x)
:=u_{K_*}(t,x)-v_K(t,x),
\label{eq:numerical-error-field}
\end{equation}
and regard it as an approximation of the infinite-dimensional error $\mE_\infty^{(K)}$ in~\eqref{eq:physical-error-full}. We use $K\in\{16,32,64\}$ for Kawahara and $K\in\{32,48,64\}$ for Nikolaevskiy.

The master coefficients are integrated with the classical fourth-order Runge--Kutta (RK4) method. In the slave, the deterministic drift, including the coupling term, is advanced with the same Runge--Kutta method, while the additive stochastic term is simulated with a Heun predictor--corrector step. The same Wiener increment is used in the predictor and corrector. All experiments run through $2\times10^4$ uniform time steps. Thus $\Delta t=4\times10^{-3}$ for the Kawahara horizon $T=80$, whereas $\Delta t=1.5\times10^{-3}$ for the Nikolaevskiy horizon $T=30$.

%
\subsection{An illustrative simulation}

Figure~\ref{fig:numerical-fields} shows a realization of the master and slave fields with slave order $K=32$. The observations are available on a uniform spatial grid with spacing $3/2$. For the Kawahara model, we take $d=0.25$ and $\sigma^2=0.25$, while for Nikolaevskiy we take $d=4$ and $\sigma^2=0.5$. The observed field is the high-resolution master $u_{K_*}$.

\begin{figure}[t]
\centering
\begin{minipage}[t]{0.24\linewidth}
\centering
\includegraphics[width=\linewidth]{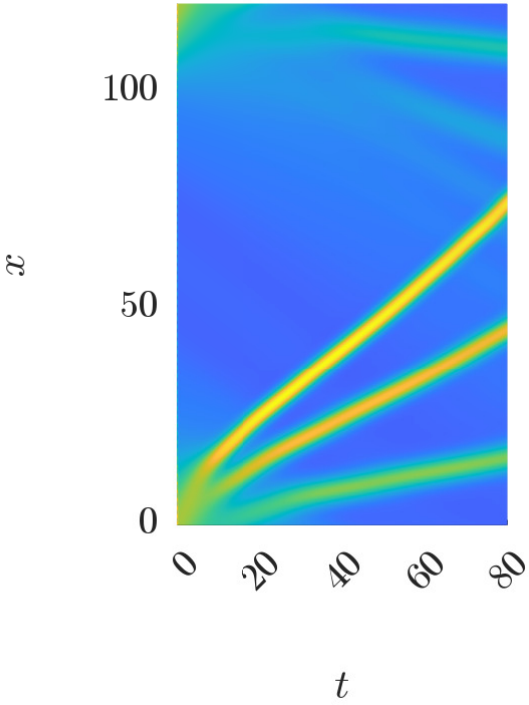}
\par\smallskip (a)
\end{minipage}\hfill
\begin{minipage}[t]{0.24\linewidth}
\centering
\includegraphics[width=\linewidth]{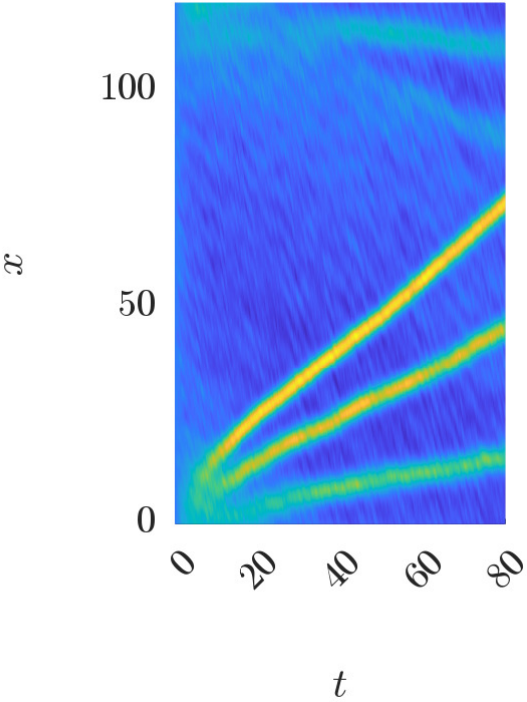}
\par\smallskip (b)
\end{minipage}\hfill
\begin{minipage}[t]{0.24\linewidth}
\centering
\includegraphics[width=\linewidth]{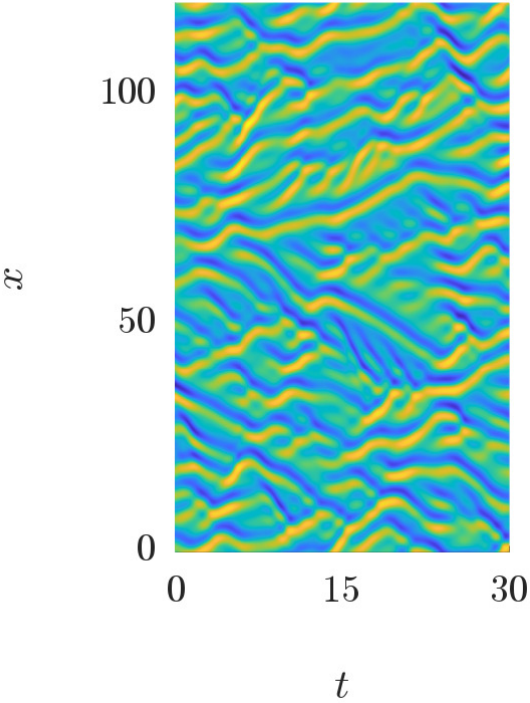}
\par\smallskip (c)
\end{minipage}\hfill
\begin{minipage}[t]{0.24\linewidth}
\centering
\includegraphics[width=\linewidth]{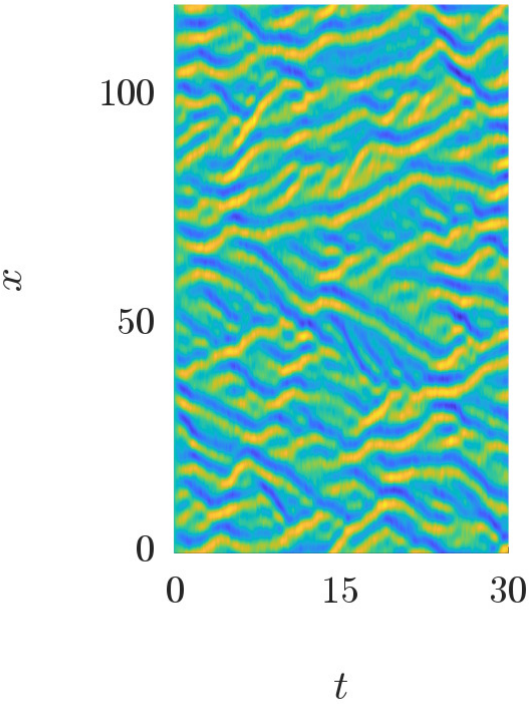}
\par\smallskip (d)
\end{minipage}
\caption{Representative master and slave fields with $K_*=64$ and $K=32$:
(a) Kawahara master, (b) Kawahara slave, (c) Nikolaevskiy master, and
(d) Nikolaevskiy slave. The Kawahara parameters are $d=0.25$ and
$\sigma^2=0.25$; the Nikolaevskiy parameters are $d=4$ and $\sigma^2=0.5$.}
\label{fig:numerical-fields}
\end{figure}

For each realization, we report the normalized mean-square synchronization error
\begin{equation}
\operatorname{NMSE}_K(t)
:=
\frac{\displaystyle\int_0^X
\left|\mE_{K_*}^{(K)}(t,x)\right|^2\,{\rm d}x}
{\displaystyle\int_0^X|u_{K_*}(t,x)|^2\,{\rm d}x}.
\label{eq:numerical-nmse}
\end{equation}
The corresponding NMSE trajectories are shown in Figure~\ref{fig:numerical-nmse}. In both examples, the slave rapidly reproduces the coherent structures of the master. Observation noise remains visible as fine-scale random fluctuations, so exact pathwise synchronization is not attained. Nevertheless, after a short transient, the NMSE remains below $10^{-1}$, i.e., the mean error power is less than ten percent of the master-field power.

\begin{figure}[t]
\centering
\begin{minipage}[t]{0.47\linewidth}
\centering
\includegraphics[width=0.6\linewidth]{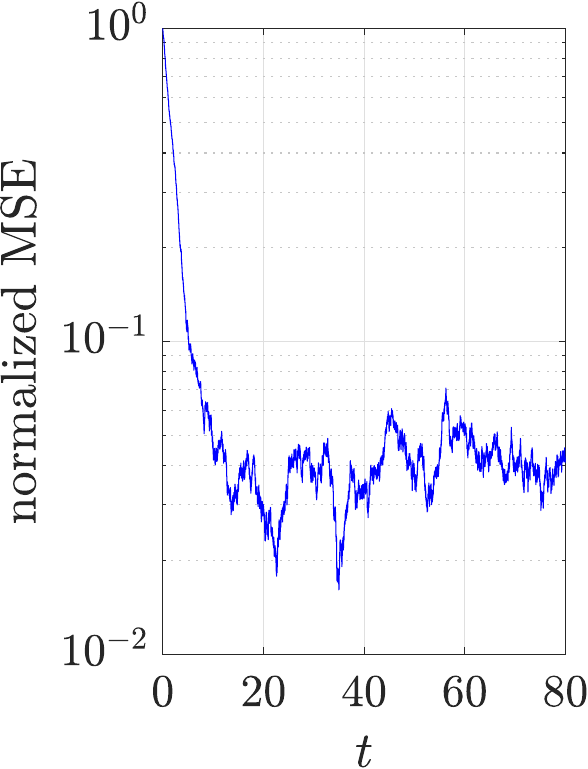}
\par\smallskip (a)
\end{minipage}\hfill
\begin{minipage}[t]{0.47\linewidth}
\centering
\includegraphics[width=0.6\linewidth]{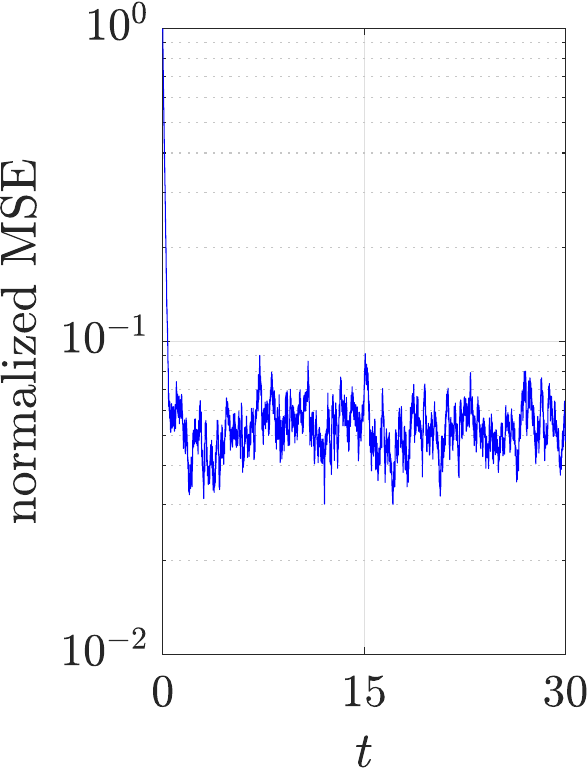}
\par\smallskip (b)
\end{minipage}
\caption{Normalized mean-square synchronization error~\eqref{eq:numerical-nmse}
for (a) the Kawahara realization and (b) the Nikolaevskiy realization shown
in Figure~\ref{fig:numerical-fields}.}
\label{fig:numerical-nmse}
\end{figure}

%
\subsection{Noise scaling and error floor}

We next consider the observation-noise variances
\[
\sigma^2\in\{10^{-3},5\times10^{-3},10^{-2},5\times10^{-2},10^{-1},
5\times10^{-1},1\}.
\]
To retain the same observation locations for every slave order, including $K=64$, we use a uniform spatial grid with spacing $3/4$; the coarser grid used above would make the order-$64$ reconstruction matrix rank deficient. For each value of $\sigma^2$, the error is averaged over $N_R=20$ independent realizations and over the final $4000$ time samples. Equivalently, with $T_0=0.8T$, the plotted statistic approximates
\begin{equation}
\widehat{\mE}_K(\sigma^2)
:=
\frac{1}{N_R}\sum_{r=1}^{N_R}
\frac{1}{X(T-T_0)}
\int_{T_0}^{T}\int_0^X
\left|\mE_{K_*,r}^{(K)}(t,x)\right|^2
\,{\rm d}x\,{\rm d}t.
\label{eq:numerical-time-ensemble-mse}
\end{equation}
The late-time window removes the synchronization transient and isolates the noise-dependent residual and the truncation floor.

Motivated by Proposition~\ref{prop:physical-error-uniform-K}, we fit the affine model
\begin{equation}
\widehat{\mE}_K(\sigma^2)
\approx \alpha_K\sigma^2+\beta_K,
\qquad \beta_K\geq0,
\label{eq:numerical-affine-fit}
\end{equation}
by least squares. Here $\alpha_K\sigma^2$ represents the stochastic contribution and $\beta_K$ estimates the error floor produced by resolving fewer modes in the slave than in the reference master. The resulting coefficients are listed in Table~\ref{tab:numerical-affine-fits}. The data and fits are shown in Figure~\ref{fig:numerical-noise-scaling}.

When $K=K_*=64$, there is no spectral mismatch relative to the numerical master and the constrained fits yield $\beta_{64}=0$. Separate power-law fits for this case give exponents $1.002$ for Kawahara and $1.001$ for Nikolaevskiy, providing direct numerical evidence that the stochastic part of the mean-square error is proportional to $\sigma^2$. For $K<K_*$, the curves flatten as $\sigma^2$ decreases and reveal a positive truncation floor. This floor is very pronounced for $K=16$ in the Kawahara experiment and for $K=32$ in the Nikolaevskiy experiment, but it has already fallen to $4.07\times10^{-5}$ at $K=32$ for Kawahara and to $1.80\times10^{-3}$ at $K=48$ for the Nikolaevskiy model. These experiments therefore reproduce the error decomposition predicted by the analysis: an approximately linear stochastic term in $\sigma^2$, plus a resolution-dependent contribution that vanishes as the slave resolution approaches that of the master. Since $K_*=64$ is itself a finite Fourier approximation, the experiment demonstrates convergence relative to the high-resolution numerical reference rather than proving convergence to the exact infinite-dimensional solution.

\begin{table}[htbp]
\caption{Constrained affine fits in~\eqref{eq:numerical-affine-fit}.}
\label{tab:numerical-affine-fits}
\centering
\begin{tabular}{lccc}
\toprule
Model & $K$ & $\alpha_K$ & $\beta_K$ \\
\midrule
Kawahara & 16 & $2.87\times10^{-2}$ & $2.75\times10^{-2}$ \\
\phantom{Kawahara} & 32 & $5.20\times10^{-2}$ & $4.07\times10^{-5}$ \\
\phantom{Kawahara} & 64 & $5.83\times10^{-2}$ & $0$ \\
Nikolaevskiy & 32 & $9.28\times10^{-1}$ & $4.46\times10^{-1}$ \\
\phantom{Nikolaevskiy} & 48 & $9.28\times10^{-1}$ & $1.80\times10^{-3}$ \\
\phantom{Nikolaevskiy} & 64 & $9.54\times10^{-1}$ & $0$
\\
\bottomrule
\end{tabular}
\end{table}

\begin{figure}[t]
\centering
\begin{minipage}[t]{0.47\linewidth}
\centering
\includegraphics[width=0.6\linewidth]{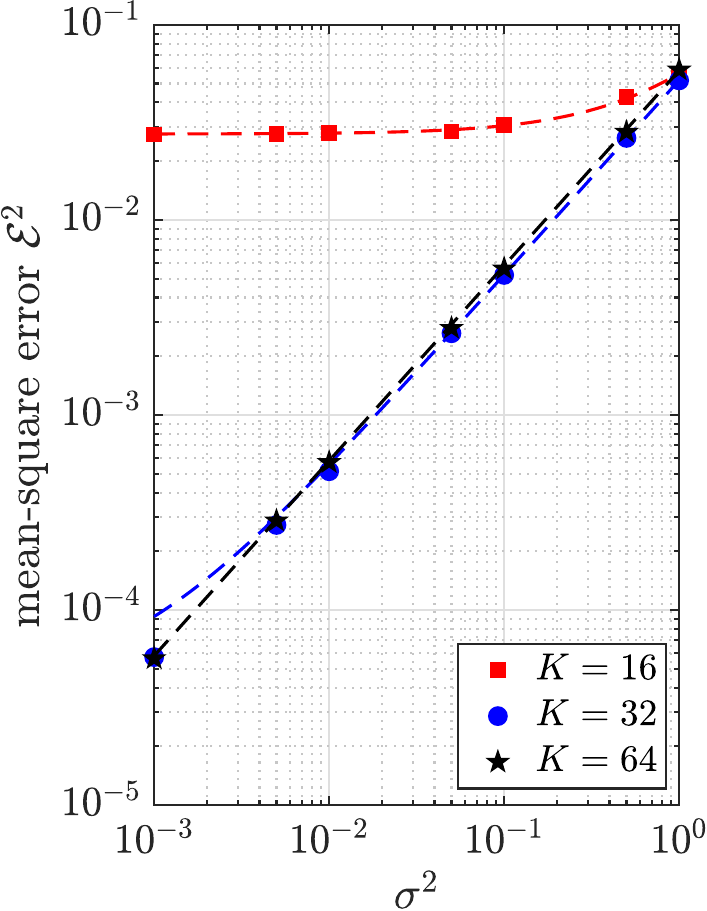}
\par\smallskip (a)
\end{minipage}\hfill
\begin{minipage}[t]{0.47\linewidth}
\centering
\includegraphics[width=0.6\linewidth]{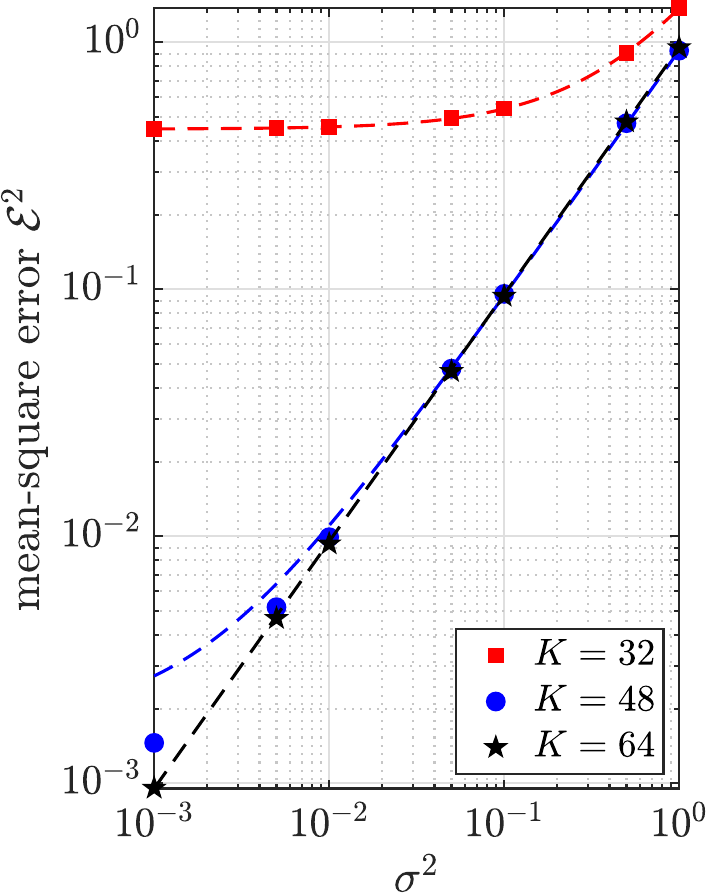}
\par\smallskip (b)
\end{minipage}
\caption{Late-time, ensemble-averaged physical-space MSE versus observation-noise
variance for (a) Kawahara and (b) Nikolaevskiy. Markers show Monte Carlo
estimates and dashed curves show the constrained affine fits
$\alpha_K\sigma^2+\beta_K$.}
\label{fig:numerical-noise-scaling}
\end{figure}

%
\section{Summary and discussion}
\label{sConclusions}

%
\subsection{Summary of contributions}

We have established a framework for assessing the robustness of master--slave synchronization when the coupling signals are corrupted by observational noise. The analysis covers a broad family of nonlinear evolution equations (including Burgers, Kuramoto--Sivashinsky, Kawahara, Benney--Lin, and Nikolaevskiy types) and exploits their shared mathematical structure: a quadratic convective term coupled to an arbitrary polynomial linear operator.

The key to the analysis is the representation of the PDE by a truncated Fourier series, in which only the $K$ first harmonics are retained. This yields a finite-dimensional model (of real dimension $2K+1$) that enables a rigorous stability analysis of the coupled slave system, both with noiseless and noisy observations. When the coupling is corrupted by white Gaussian observational noise, the synchronization error decomposes into four physically interpretable components: the unresolved Fourier tail (spectral truncation error), the resolved Galerkin error, the deterministic reference error of the noiseless slave, and the stochastic deviation induced by the noise. Propositions~\ref{prop:stochastic-closeness} and~\ref{prop:global-stochastic-closeness} show that the stochastic deviation is controlled by the noise intensity $\sigma$ with an \(\mathcal O(\sigma^2)\) mean-square bound on finite horizons, and remains uniformly bounded in time when the dissipativity gap $\mu_0$ is positive. Proposition~\ref{prop:physical-error-uniform-K} combines these stochastic bounds with spectral approximation estimates to quantify the total physical-space synchronization error relative to the infinite-dimensional master field. Intuitively, if the coupling has sufficient strength, the error does not grow uncontrollably despite the persistent noise injection, the steady-state mean-square error scales quadratically with the noise level and, for sufficiently regular models, the truncation error vanishes as the slave resolution increases. In the high-resolution limit, the dominant residual is the noise-dependent term $\mathcal O(\sigma^2)$, as illustrated by the numerical experiments in Section~\ref{sec:numerical-results}.

The coupling scheme analyzed here is mathematically equivalent to continuous data assimilation (CDA) by nudging \cite{BroeckerDelSartoHieberPalmaZoechling2026}, but formulated in Fourier space rather than physical space. This distinction is significant because (i) the spectral formulation makes the dependence on the observation reconstruction matrix $R_K$ explicit, allowing the Hilbert--Schmidt norm $\|R_K\|_{K;L^2,F}$ to enter the error bounds directly; (ii) the analysis accommodates purely dispersive equations (such as the Kawahara model) that fall outside the typical scope of parabolic CDA frameworks; and (iii) the finite-$K$ stability criteria do not require global coercivity of the linear operator, only a sufficiently strong scalar gain $d$.

%
\subsection{Limitations and future work}

The obtained results rely on several assumptions that warrant discussion.
\begin{itemize}
\item Sobolev regularity $s>\frac{3}{2}$ ensures uniform $L^\infty$ bounds on solutions and their derivatives via the embedding $H^s_{\rm per} \hookrightarrow W^{1,\infty}$. For dissipative members of the family (e.g., Kuramoto-Sivashinsky, Benney--Lin), parabolic regularization makes this assumption natural. For purely dispersive cases, we expect that well-posedness analysis and smooth initial data can justify it on finite horizons.
\item Uniform-in-$K$ bounds on master, reference and Galerkin trajectories are required for the high-resolution ($K\to\infty$) limit. These bounds hold for dissipative systems with absorbing sets but extending the theory to dispersive cases may require additional structural assumptions.
\item Exact reconstruction $R_K u_N(t) = \bar a(t)$ presumes that observations are collected from the truncated master field. If the full infinite-dimensional field is observed, unresolved modes can alias into the reconstruction, requiring anti-aliasing filters or denser grids.
\end{itemize}

Several directions for future research arise. An obvious venue for exploration is the extension to 2D and 3D spatial models. Another significant extension is the incorporation of adaptive parameter estimation into the stochastic synchronization framework, along the lines of the adaptive scheme proposed in \cite{Miguez2024} for the Kuramoto-Sivashinsky model. Combining simultaneous state synchronization and parameter learning with the stochastic analysis presented here would provide a mathematically grounded framework for data assimilation in partially-known nonlinear dissipative or dispersive systems.

\section*{Declaration of generative AI and AI-assisted technologies in the manuscript preparation process}

During the preparation of this work the authors used ChatGPT and Lumo AI to assist with language editing, literature searches and preliminary checking of several mathematical derivations. After using these tools, the authors reviewed and edited the content as needed and take full responsibility for the content of the published article.


\appendix

%
\section{Proof of Proposition~\ref{prop:local-sync-general}}
\label{app:proof-prop1}

Under Assumption~\ref{ass:exact-reconstruction}, the retained master coefficients $\bar a_k(t)$ are exactly available to the slave. The master and slave vector fields are given by \eqref{eq:master-vector-general} and \eqref{eq:slave-vector-general} and subtracting them yields the error equation \eqref{eq:error-ode-general}.

For \(z,h\in\mathcal H_K\), extend the coefficients to negative indices by conjugate symmetry and set them equal to zero outside \([-K,K]\). The real-linear operator \(Q(z)\) acts according to \eqref{eq-Qzh}, namely,
\begin{equation}
\label{eq:Q-def}
\bigl[Q(z)h\bigr]_k
:=
k\sum_{\ell=-K}^{K}
\left(z_\ell h_{k-\ell}+h_\ell z_{k-\ell}\right),
\qquad k=0,\ldots,K.
\end{equation}
In the \(2K+1\) real coordinates of \eqref{eq-2K+1}, this is an ordinary real matrix whose entries depend linearly on the real and imaginary parts of \(z\). Assumption~\ref{ass:bounded-master} therefore implies that \(Q(\bar a(t))\) is uniformly bounded on \([0,T]\).

Recalling \eqref{eq:linearization-general}, we let \(A(t)=M(t)-dI\), where
\[
M(t)=\Lambda(P,\omega_0)-\frac{ic_1\omega_0}{2}Q(\bar a(t)).
\]
Since $Q(\bar a(t))$ is uniformly bounded (and $\Lambda(P,\omega_0)$ is constant), it follows that \(M(t)\) is uniformly bounded on \([0,T]\) as well. Hence, we can construct
\[
m_T:=
\sup_{\substack{0\le t\le T\\e\in\mathcal H_K, ~~e \ne 0}}
\frac{(e,M(t)e)_{\mathcal H_K}}{\|e\|^2}
<\infty
\]
and set \(d_0:=\max\{m_T,0\}\). For any \(d>d_0\), let \(\alpha:=d-m_T>0\). It follows directly that
\begin{equation}
\label{eq:uniform-dissipativity}
(e,A(t)e)_{\mathcal H_K}\le -\alpha\|e\|^2,
\quad e\in\mathcal H_K,\quad 0\le t\le T.
\end{equation}

Let us now define \(V(e)=\frac12\|e\|^2\). Along solutions of \eqref{eq:error-ode-general},
\begin{equation}
\dot V(e)
=(e,\dot e)_{\mathcal H_K}
=(e,A(t)e)_{\mathcal H_K}+
(e,r(e))_{\mathcal H_K} 
\le -\alpha\|e\|^2+C_r\|e\|^3,
\label{eq:Vdot-general}
\end{equation}
where we have used the bound $\|r(e)\| \le C_r \| e \|^2$ in \eqref{eq:remainder-bound} and a Cauchy-Schwarz inequality. If we choose \(\rho>0\) such that \(C_r\rho\le \alpha/2\) then, whenever \(\|e\|\le\rho\),
\[
\dot V(e)\le -\frac{\alpha}{2}\|e\|^2=-\alpha V(e).
\]
If \(\|e(0)\|<\rho\), this inequality prevents the solution from leaving the ball $\sB(0,\rho)$ and yields
\[
\|e(t)\|\le {\rm e}^{-\alpha t/2}\|e(0)\|,
\qquad 0\le t\le T.
\]

%
\section{Proof of Lemma \ref{lem:ref-dissipativity}}
\label{ap:ref-dissipativity}

From \eqref{eq:jacobian-ref-general} and Assumption \ref{ass:scalar-coupling}, the linearization along the reference path is
\[
J_{\mathrm{ref}}(t):=\partial_e {\sf f}(t, e_{\mathrm{ref}}(t)) 
= \Lambda(P,\omega_0)-dI-\frac{i c_1\omega_0}{2}Q\bigl(\bar a(t)-e_{\mathrm{ref}}(t)\bigr).
\]

Let us define the operator
\[
M_{\mathrm{ref},d}(t):=J_{\mathrm{ref}}(t) + dI
= \Lambda(P,\omega_0)-\frac{i c_1\omega_0}{2}Q\bigl(\bar a(t)-e_{\mathrm{ref}}(t)\bigr),
\]
that may depend on the coupling gain $d$ through $e_{\rm ref}$. This operator is bounded. To see it, let us note that the operator $Q(z)$ in \eqref{eq:Q-def} depends linearly on $\operatorname{Re}(z_k)$ and $\operatorname{Im}(z_k)$. Moreover, Assumption \ref{ass:bounded-master} implies \eqref{eq:master-coeff-bound-V} and Assumption \ref{ass:reference-error-path} guarantees \eqref{eq:reference-bound}. The inequalities \eqref{eq:master-coeff-bound-V} and \eqref{eq:reference-bound} together give
\beq
\sup_{0\le t\le T}\|\bar a(t)-e_{\mathrm{ref}}(t)\|_{K;L^2}
\le C_{a,T}+C_{\mathrm{ref},T}<\infty,
\label{eq-carallada}
\eeq
so $Q(\bar a(t)-e_{\mathrm{ref}}(t))$ is uniformly bounded on $[0,T]$. Moreover, note that the master field is independent of $d$ and $C_{{\rm ref},T}$ holds for all $d > d_{\rm ref}$, hence the bound \eqref{eq-carallada} holds for all $d > d_{\rm ref}$, where $d_{\rm ref}$ is the reference coupling gain in Assumption \ref{ass:reference-error-path}. Since $\Lambda(P,\omega_0)$ is constant, the linear operator $M_{\mathrm{ref},d}(t)$ is uniformly bounded on $\mH_K$ w.r.t. the weighted norm $\| \cdot \|_{K;L^2}$ as well.

At fixed $K$, the weighted norm $\|\cdot\|_{K;L^2}$ is equivalent to the Euclidean norm, so the boundedness of $M_{\mathrm{ref},d}(t)$ controls its weighted quadratic form and we can define
\[
m_*:=\sup_{\substack{0\le t\le T\\ d>d_{\rm ref}\\ z\in\mathcal H_K,\,z\ne0}}
\frac{\langle z,M_{\mathrm{ref},d}(t)z\rangle_{K;L^2}}{\|z\|_{K;L^2}^2}
<\infty.
\]
Set $d_0:=\max\{m_*,d_{\mathrm{ref}}\}$ and, for any $d>d_0$, let $\alpha:=d-m_*>0$. Then
\[
\langle z,J_{\mathrm{ref}}(t)z\rangle_{K;L^2}
=
\langle z,M_{\mathrm{ref},d}(t)z\rangle_{K;L^2}-d\|z\|_{K;L^2}^2
\le -\alpha\|z\|_{K;L^2}^2
\]
for all $z\in\mathcal H_K$ and $t\in[0,T]$, which proves \eqref{eq:ref-dissipativity}.

%
\section{Proof of Proposition \ref{prop:stochastic-closeness}}
\label{ap:prop:stochastic-closeness}

In the $2K+1$ real coordinates of \eqref{eq-2K+1}, the drift $\Delta\mapsto{\sf f}(t,e_{\rm ref}(t)+\Delta)-{\sf f}(t,e_{\rm ref}(t))$ is a polynomial of degree two, hence locally Lipschitz, with a Lipschitz constant uniform in $t\in[0,T]$ (because $t\mapsto J_{\rm ref}(t)$ is continuous on the compact interval $[0,T]$ and the diffusion coefficient $-\sigma dR_K$ is constant). By the standard localization theorem for SDEs with locally Lipschitz coefficients (see, e.g., \cite[Chapter IV]{Ikeda89}), there is a unique maximal strong solution $(\Delta,\tau_e)$, defined up to an explosion time $\tau_e$. The stopping-time argument below shows that $\tau_e>T$ almost surely.

Let us denote
\[
V(\Delta):=\|\Delta\|_{K;L^2}^2
\quad \text{and} \quad
S:=\sigma^2d^2\|R_K\|_{K;L^2,F}^2.
\]
It\^o's formula, combined with the inequality \eqref{eq:ref-dissipativity} and the exact cancellation \eqref{eq:nonlinear-energy-cancellation}, give
\begin{align}
{\rm d}V(\Delta(t))
&=\left(2\langle\Delta(t),J_{\mathrm{ref}}(t)\Delta(t)\rangle_{K;L^2}+S\right){\rm d}t
-2\sigma d\,\langle\Delta(t),R_K\,{\rm d}W_N(t)\rangle_{K;L^2}
\nonumber\\
&\le\left(-2\alpha V(\Delta(t))+S\right){\rm d}t
-2\sigma d\,\langle\Delta(t),R_K\,{\rm d}W_N(t)\rangle_{K;L^2}
\label{eq:Ito-Delta-est}
\end{align}
up to the explosion time. 

In the sequel, we use stopping times only to justify non-explosion and the moment estimates. Let \(\tau_n\) be the first time the local solution reaches \(\|\Delta\|_{K;L^2}=n\), with \(\inf\emptyset=\infty\). Dropping the negative drift term in \eqref{eq:Ito-Delta-est} gives
\beq
V(\Delta(t\wedge\tau_n))\le St+M_t^{(n)}
\quad \text{and} \quad
M_t^{(n)}:=-2\sigma d\int_0^{t\wedge\tau_n}
\langle\Delta(s),R_K\,{\rm d}W_N(s)\rangle_{K;L^2}.
\label{eq---A}
\eeq
The stochastic term $M_t^{(n)}$ above is a square-integrable martingale. To see this, write the integrand component-wise, with $\{\beta_j\}_{j=1}^N$ the canonical basis of $\mathbb R^N$, to arrive at
\[
\langle\Delta(s),R_K{\rm d}W_N(s)\rangle_{K;L^2}
=\sum_{j=1}^N \underbrace{\langle\Delta(s),R_K\beta_j\rangle_{K;L^2}}_{=:g_j(s)}\,{\rm d}W_N^{(j)}(s).
\]
For $s\le t\wedge\tau_n$ we have $\|\Delta(s)\|_{K;L^2}\le n$, so the integrand processes $g_j(s)$ are adapted, and bounded by Cauchy--Schwarz and \eqref{eq:induced-HS-norm}, namely,
\[
\sum_{j=1}^N |g_j(s)|^2
\le \|\Delta(s)\|_{K;L^2}^2\|R_K\|_{K;L^2,F}^2
\le n^2\|R_K\|_{K;L^2,F}^2<\infty.
\]
Hence the integrand lies in $L^2(\Omega\times[0,t])$, $M_t^{(n)}$ is a true martingale (not merely local), and Itô's isometry applies in the sequel.

Let us define
\[
Y_n(t):=\mbE\left[\sup_{0\le s\le t}
\|\Delta(s\wedge\tau_n)\|_{K;L^2}^2\right].
\]
Doob's \(L^2\) maximal inequality and It\^o's isometry yield
\beq
\mbE\sup_{0\le s\le t}|M_s^{(n)}|
\le2\bigl(\mbE|M_t^{(n)}|^2\bigr)^{1/2}
\le4\sqrt S\left(\int_0^tY_n(s)\,{\rm d}s\right)^{1/2}.
\label{eq---B}
\eeq
Moreover, for \(0\le t\le T\), the combination of \eqref{eq---A} and \eqref{eq---B}, followed by an application of Young's inequality gives
\[
Y_n(t)\le ST+4\sqrt S\left(\int_0^tY_n(s)\,{\rm d}s\right)^{1/2}
\le S(T+8)+\frac12\int_0^tY_n(s)\,{\rm d}s.
\]
Gronwall's inequality in integral form now readily leads to
\beq
Y_n(T)\le C_TS, \quad \text{with} \quad C_T:=(T+8){\rm e}^{T/2}.
\label{eq-indepe}
\eeq
This estimate is independent of \(n\). Since continuity gives \(V(\Delta(\tau_n))=n^2\) on \(\{\tau_n\le T\}\), we observe that
\beq
Y_n(T) \ge \mbE\bigl[\|\Delta(\tau_n\wedge T)\|^2_{K;L^2}\,\mathbf 1_{\{\tau_n\le T\}}\bigr] = n^2\,\mathbb P(\tau_n\le T)
\nn
\eeq
which implies
\[
\mathbb P(\tau_n\le T)\le\frac{C_TS}{n^2} \to 0 \quad \text{as $n\to\infty$}.
\]
Thus, the solution is non-explosive on \([0,T]\), and $\Delta(s\wedge\tau_n)\to\Delta(s)$ almost surely when $n\to\infty$, for every $s$. Since the $\sup$ of point-wise limits is dominated by the limit of the suprema, we can write
\beq
\mbE\left[
	\sup_{0\le s\le T} \|\Delta(s)\|^2_{K;L^2}
\right]
\le \mbE\left[\lim\inf_{n\to\infty}\sup_{0\le s\le T}\|\Delta(s\wedge\tau_n)\|^2_{K;L^2}\right]
\le \liminf_{n\to\infty} Y_n(T) \le C_T\,S,
\nn
\eeq
where the second inequality is granted by Fatou's lemma and the third one is due to the independence of $n$ in \eqref{eq-indepe}. This gives \eqref{eq:stochastic-closeness-bound}, and Markov's inequality leads to \eqref{eq:Delta-prob-bound}.

Finally, the finite supremum moment makes the unstopped stochastic integral in \eqref{eq:Ito-Delta-est} square-integrable on \([0,T]\), with zero mean. Taking expectations yields
\[
\frac{{\rm d}}{{\rm d}t}\mbE V(\Delta(t))
\le-2\alpha\,\mbE V(\Delta(t))+S
\]
for almost every \(t\in[0,T]\). Since \(V(\Delta(0))=0\), integration proves \eqref{eq-xx1}.

%
\section{Proof of Lemma \ref{lem:global-dissipativity-fourier}}
\label{ap:lem:global-dissipativity-fourier}

%
\subsection{The proof}

Recall \(V_K\) and the orthogonal projection \({\rm P}_K\) from Section~\ref{sec:finite-horizon-closeness}. For fixed \(t\) and \(\Delta\), set \(h_\Delta:=\mathcal F^{-1}\Delta\) and \(G(t,\Delta):={\sf f}(t,e_{\mathrm{ref}}+\Delta)-{\sf f}(t,e_{\mathrm{ref}})\). The physical-space drift difference is
\begin{equation}
\mathcal F^{-1}G(t,\Delta)
=-c_1{\rm P}_K\bigl(v_{\mathrm{ref},K}\partial_xh_\Delta
+h_\Delta\partial_x v_{\mathrm{ref},K}-h_\Delta\partial_x h_\Delta\bigr)\\
-P(\partial_x)h_\Delta-dh_\Delta,
\label{eq-ddrift-phys}
\end{equation}
as shown in Appendix~\ref{ap-eq-ddrift-phys} below. Since \(h_\Delta\in V_K\), the projector can be removed after taking the inner product with \(h_\Delta\). The cancellation \eqref{eq:nonlinear-energy-cancellation} and periodic integration by parts give
\beq
\langle\Delta,G(t,\Delta)\rangle_{K;L^2}
=-\frac{c_1}{2}\int_0^X(\partial_x v_{\mathrm{ref},K})h_\Delta^2\,{\rm d}x
-\operatorname{Re}\langle h_\Delta,P(\partial_x)h_\Delta\rangle_{L^2(0,X)}
-d\|h_\Delta\|_{L^2(0,X)}^2.
\label{eq:global-dissipativity-proof-start}
\eeq
The first term is at most \(\frac{|c_1|}{2}L_{\mathrm{ref},K}\|h_\Delta\|_{L^2}^2\), where $L_{{\rm ref},K}<\infty$ when \eqref{eq:ref-slave-Linfty-bound} holds. For the linear part, Parseval's identity yields
\begin{equation}
\label{eq:global-linear-bound-fourier}
-\operatorname{Re}\langle h_\Delta,P(\partial_x)h_\Delta\rangle_{L^2(0,X)}
=X\sum_{k=-K}^K\mathfrak q(\omega_0k)|\Delta_k|^2
\le\mathfrak q_*\|h_\Delta\|_{L^2(0,X)}^2,
\end{equation}
where \eqref{eq:qstar-bounded} yields $\mathfrak q_*<\infty$. Combining these bounds proves \eqref{eq:global-dissipativity}. The final assertion follows from \eqref{eq:Delta-linearization} and \eqref{eq:nonlinear-energy-cancellation}.

%
\subsection{Derivation of equation \eqref{eq-ddrift-phys}}
\label{ap-eq-ddrift-phys}

Recall that, when \(D=dI\), the drift of the stochastic error equation is
\[
{\sf f}(t,e)=
(\Lambda(P,\omega_0)-dI)e
-\frac{i c_1\omega_0}{2}
\Bigl(\eta(\bar a(t))-\eta(\bar a(t)-e)\Bigr).
\]
{
Fix \(t\in[0,T]\) and \(\Delta\in\mathcal H_K\) and let
$
G(t,\Delta):={\sf f}(t,e_{\mathrm{ref}}(t)+\Delta)-{\sf f}(t,e_{\mathrm{ref}}(t)).
$
}
By subtracting the two drift values, the terms \(\eta(\bar a(t))\) cancel and we obtain
\begin{equation}
G(t,\Delta) =(\Lambda(P,\omega_0)-dI)\Delta-\frac{i c_1\omega_0}{2}
\Bigl(
\eta(\bar a(t)-e_{\mathrm{ref}}(t))
-
\eta(\bar a(t)-e_{\mathrm{ref}}(t)-\Delta)
\Bigr).
\label{eq:appendix-G-drift-difference}
\end{equation}
The linear part is immediate from the definition of the Fourier multiplier; we obtain
\beq
\mathcal F^{-1}\bigl[\Lambda(P,\omega_0)\Delta\bigr]
=-P(\partial_x)h_\Delta
\quad \text{and} \quad
\mathcal F^{-1}[-d\Delta]=-d h_\Delta.
\label{eq-lin-apC}
\eeq

It remains to identify the nonlinear part. For any Hermitian-symmetric coefficient vector \(z_{-K:K}\), let \(w=\mathcal F^{-1}z\in V_K\). From the definition of the convolution map \(\eta\),
$
\eta_k(z)=k\sum_{\ell=-K}^{K}z_\ell z_{k-\ell}
$, with the usual convention that $z_k=0$ for $|k|>K$. Since \(\eta\) retains only the modes \(|k|\le K\), while \(w^2\) may contain modes up to order \(2K\), the correct Fourier identity is
\begin{equation}
\label{eq:appendix-nonlinear-fourier-identity}
\mathcal F^{-1}\left[\frac{i\omega_0}{2}\eta(z)\right]
=
\frac12{\rm P}_K\partial_x(w^2)={\rm P}_K(w\partial_x w).
\end{equation}
The physical field associated with \(\bar a(t)-e_{\mathrm{ref}}(t)\) is
\[
v_{\mathrm{ref},K}(t,x)=\mathcal F^{-1}\bigl(\bar a(t)-e_{\mathrm{ref}}(t)\bigr)(x),
\]
while the physical field associated with \(\bar a(t)-e_{\mathrm{ref}}(t)-\Delta\) is
\[
v_{\mathrm{ref},K}(t,x)-h_\Delta(t,x).
\]
Applying \eqref{eq:appendix-nonlinear-fourier-identity} to the nonlinear term in \eqref{eq:appendix-G-drift-difference} yields
\begin{align*}
\mathcal F^{-1}\left[
-\frac{i c_1\omega_0}{2}
\Bigl(
\eta(\bar a-e_{\mathrm{ref}})
-
\eta(\bar a-e_{\mathrm{ref}}-\Delta)
\Bigr)
\right]
&=
-c_1{\rm P}_K\Bigl[
 v_{\mathrm{ref},K}\partial_x v_{\mathrm{ref},K}
\\
&\quad -
 (v_{\mathrm{ref},K}-h_\Delta)\partial_x(v_{\mathrm{ref},K}-h_\Delta)
\Bigr].
\end{align*}
Expanding the second product on the right-hand side above we have
\begin{align*}
(v_{\mathrm{ref},K}-h_\Delta)\partial_x(v_{\mathrm{ref},K}-h_\Delta)
&=(v_{\mathrm{ref},K}-h_\Delta)(\partial_xv_{\mathrm{ref},K}-\partial_xh_\Delta)
\\
&=v_{\mathrm{ref},K}\partial_xv_{\mathrm{ref},K}
-v_{\mathrm{ref},K}\partial_xh_\Delta -h_\Delta\partial_xv_{\mathrm{ref},K}
+h_\Delta\partial_xh_\Delta
\end{align*}
and, therefore,
\beq
 v_{\mathrm{ref},K}\partial_x v_{\mathrm{ref},K}
-
 (v_{\mathrm{ref},K}-h_\Delta)\partial_x(v_{\mathrm{ref},K}-h_\Delta)
=
v_{\mathrm{ref},K}\partial_xh_\Delta
+h_\Delta\partial_xv_{\mathrm{ref},K}
-h_\Delta\partial_xh_\Delta.
\label{eq-nonlin-appC}
\eeq
Combining \eqref{eq-nonlin-appC} with \eqref{eq-lin-apC} yields Eq. \eqref{eq-ddrift-phys}. The projector can subsequently be removed only after taking the \(L^2(0,X)\) inner product with \(h_\Delta\in V_K\).

%
\section{Proof of Lemma~\ref{lem:uniform-Galerkin-error}}
\label{ap:Galerkin-error-Hs}

We prove Lemma~\ref{lem:uniform-Galerkin-error} by an energy argument. Recall from Section~\ref{sec:finite-horizon-closeness} that
\[
V_K:=\operatorname{span}\{ {\rm e}^{ik\omega_0x}: |k|\le K\}
\]
is the order-\(K\) Fourier--Galerkin space and that ${\rm P}_K:L^2(0,X)\to V_K$ is the orthogonal Fourier projection. Recall the notation $\zeta_K(t,\cdot):=u(t,\cdot)-{\rm P}_K u(t,\cdot)$ and $\gamma_K(t,\cdot):={\rm P}_K u(t,\cdot)-u_K(t,\cdot)$, which yields $u-u_K=\zeta_K+\gamma_K$. The field $\zeta_K$ contains the unresolved modes, while $\gamma_K\in V_K$ measures the difference between the projection of the full solution and the solution of the closed Fourier--Galerkin system.

Because ${\rm P}_K$ is independent of time and commutes with the
constant-coefficient operator $P(\partial_x)$, applying ${\rm P}_K$ to the
full PDE gives
\begin{equation}
\partial_t {\rm P}_K u + P(\partial_x){\rm P}_K u
+c_1{\rm P}_K(u\partial_xu)=0.
\label{eq:appendix-projected-full-PDE}
\end{equation}
The order-$K$ Fourier--Galerkin approximation satisfies
\begin{equation}
\partial_t u_K+P(\partial_x)u_K
+c_1{\rm P}_K(u_K\partial_xu_K)=0.
\label{eq:appendix-Galerkin-PDE}
\end{equation}
Subtracting \eqref{eq:appendix-Galerkin-PDE} from
\eqref{eq:appendix-projected-full-PDE} yields the Galerkin-error equation
\begin{equation}
\partial_t\gamma_K+P(\partial_x)\gamma_K
+c_1{\rm P}_K\bigl(u\partial_xu-u_K\partial_xu_K\bigr)=0.
\label{eq:appendix-Galerkin-error-equation}
\end{equation}

We next take the real $L^2(0,X)$ inner product of
\eqref{eq:appendix-Galerkin-error-equation} with $\gamma_K$. Since
\({\rm P}_K\) is the orthogonal projection onto \(V_K\) and \(\gamma_K\in V_K\),
\[
\langle\gamma_K,{\rm P}_Kf\rangle_{L^2(0,X)}
=\langle\gamma_K,f\rangle_{L^2(0,X)}.
\]
The projector can therefore be removed from the nonlinear term in the
energy identity. We obtain
\begin{equation}
\frac12\frac{{\rm d}}{{\rm d}t}
\|\gamma_K(t,\cdot)\|_{L^2(0,X)}^2
=-\operatorname{Re}\langle\gamma_K,
P(\partial_x)\gamma_K\rangle_{L^2(0,X)} -c_1\int_0^X\gamma_K
\bigl(u\partial_xu-u_K\partial_xu_K\bigr)\,{\rm d}x.
\label{eq:appendix-Galerkin-energy-identity}
\end{equation}

The linear term is controlled directly by its Fourier symbol. If
$\widehat\gamma_k(t)$ denotes the $k$th Fourier coefficient of $\gamma_K$,
then Parseval's identity and \eqref{eq:qstar-bounded} give
\begin{equation}
-\operatorname{Re}\langle\gamma_K,
P(\partial_x)\gamma_K\rangle_{L^2(0,X)}
=X\sum_{|k|\le K}
\operatorname{Re}\bigl[-P(i\omega_0k)\bigr]
|\widehat\gamma_k(t)|^2
\le\mathfrak q_*
\|\gamma_K(t,\cdot)\|_{L^2(0,X)}^2.
\label{eq:appendix-Galerkin-linear-bound}
\end{equation}

To estimate the nonlinear term, use $u-u_K=\zeta_K+\gamma_K$ to write
\begin{align}
u\partial_xu-u_K\partial_xu_K
&=(u-u_K)\partial_xu
+u_K\partial_x(u-u_K)
\nonumber\\
&=\zeta_K\partial_xu+\gamma_K\partial_xu
+u_K\partial_x\zeta_K+u_K\partial_x\gamma_K.
\label{eq:appendix-Galerkin-nonlinear-decomposition}
\end{align}
Each of the four resulting integrals can be bounded separately. First,
Hölder's inequality gives
\begin{align*}
\left|\int_0^X\gamma_K\zeta_K\partial_xu\,{\rm d}x\right|
&\le \|\partial_xu(t,\cdot)\|_{L^\infty(0,X)}
\|\gamma_K(t,\cdot)\|_{L^2(0,X)}
\|\zeta_K(t,\cdot)\|_{L^2(0,X)},\\
\left|\int_0^X\gamma_K^2\partial_xu\,{\rm d}x\right|
&\le \|\partial_xu(t,\cdot)\|_{L^\infty(0,X)}
\|\gamma_K(t,\cdot)\|_{L^2(0,X)}^2 \quad\text{and}\\
\left|\int_0^X\gamma_Ku_K\partial_x\zeta_K\,{\rm d}x\right|
&\le \|u_K(t,\cdot)\|_{L^\infty(0,X)}
\|\gamma_K(t,\cdot)\|_{L^2(0,X)}
\|\partial_x\zeta_K(t,\cdot)\|_{L^2(0,X)}.
\end{align*}
For the last integral, periodicity and integration by parts yield
\begin{equation}
\int_0^X\gamma_Ku_K\partial_x\gamma_K\,{\rm d}x
=\frac12\int_0^Xu_K\partial_x(\gamma_K^2)\,{\rm d}x
=-\frac12\int_0^X(\partial_xu_K)\gamma_K^2\,{\rm d}x
\nonumber
\end{equation}
and hence
\[
\left|\int_0^X\gamma_Ku_K\partial_x\gamma_K\,{\rm d}x\right|
\le\frac12\|\partial_xu_K(t,\cdot)\|_{L^\infty(0,X)}
\|\gamma_K(t,\cdot)\|_{L^2(0,X)}^2.
\]

For $s>\frac{3}{2}$, the one-dimensional Sobolev embedding
$H^s_{\rm per}(0,X)\hookrightarrow W^{1,\infty}(0,X)$ and the uniform
bounds in Assumption~\ref{ass:uniform-physical-error-regularity} imply that the
$L^\infty(0,X)$ norms appearing above are bounded independently of $K$ and
$t\in[0,T]$. More precisely, if $C_{\rm emb}=C_{\rm emb}(s,X)$ denotes the
norm of this embedding and
\[
U_{s,T}:=C_{\rm emb}M_{s,T},
\]
then
\[
\|\partial_xu(t,\cdot)\|_{L^\infty(0,X)},\quad
\|u_K(t,\cdot)\|_{L^\infty(0,X)} \quad\text{and}\quad
\|\partial_xu_K(t,\cdot)\|_{L^\infty(0,X)}
\le U_{s,T}.
\]
Combining the four estimates with
\eqref{eq:appendix-Galerkin-linear-bound}, and setting
\[
A_{s,T}:=\max\{\mathfrak q_*,0\}
+\frac{3}{2}|c_1|U_{s,T},
\qquad
B_{s,T}:=|c_1|U_{s,T},
\]
we obtain
\begin{align}
\frac12\frac{{\rm d}}{{\rm d}t}
\|\gamma_K(t,\cdot)\|_{L^2(0,X)}^2
&\le A_{s,T}\|\gamma_K(t,\cdot)\|_{L^2(0,X)}^2
\nonumber\\
&\quad+B_{s,T}\|\gamma_K(t,\cdot)\|_{L^2(0,X)}
\Bigl(
	\|\zeta_K(t,\cdot)\|_{L^2(0,X)} 
	+ \|\partial_x\zeta_K(t,\cdot)\|_{L^2(0,X)}
\Bigr).
\label{eq:appendix-Galerkin-pre-Young}
\end{align}

Young's inequality, in the form $ab\le(a^2+b^2)/2$, absorbs the factor
$\|\gamma_K(t,\cdot)\|_{L^2(0,X)}$ into the first term. Multiplying the
resulting inequality by $2$ gives
\begin{equation}
\frac{{\rm d}}{{\rm d}t}
\|\gamma_K(t,\cdot)\|_{L^2(0,X)}^2
\le C_{1,s,T}\|\gamma_K(t,\cdot)\|_{L^2(0,X)}^2
+C_{2,s,T}
\left(
\|\zeta_K(t,\cdot)\|_{L^2(0,X)}^2
+\|\partial_x\zeta_K(t,\cdot)\|_{L^2(0,X)}^2
\right),
\label{eq:appendix-Galerkin-post-Young}
\end{equation}
where
\[
C_{1,s,T}:=2A_{s,T}+1
\quad\text{and}\quad
C_{2,s,T}:=2B_{s,T}^2,
\]
and both constants are independent of $K$. By
\eqref{eq:projection-tail-Hs}, the uniform bound in
Assumption~\ref{ass:uniform-physical-error-regularity}, and
$\kappa_K\ge1$,
\begin{equation}
\|\zeta_K(t,\cdot)\|_{L^2(0,X)}^2
+\|\partial_x\zeta_K(t,\cdot)\|_{L^2(0,X)}^2
\le M_{s,T}^2 \left(
	\kappa_K^{-2s}+\kappa_K^{-2(s-1)}
\right)
\le 2M_{s,T}^2\kappa_K^{-2(s-1)}.
\end{equation}
Consequently, with
\[
C_{3,s,T}:=2C_{2,s,T}M_{s,T}^2,
\]
which is also independent of $K$,
\begin{equation}
\frac{{\rm d}}{{\rm d}t}
\|\gamma_K(t,\cdot)\|_{L^2(0,X)}^2
\le C_{1,s,T}\|\gamma_K(t,\cdot)\|_{L^2(0,X)}^2
+C_{3,s,T}\kappa_K^{-2(s-1)}.
\label{eq:appendix-Galerkin-Gronwall-input}
\end{equation}

Finally, the projected initialization
$u_K(0,\cdot)={\rm P}_Ku(0,\cdot)$ implies
$\gamma_K(0,\cdot)=0$. The integral form of Gronwall's inequality applied
to \eqref{eq:appendix-Galerkin-Gronwall-input} gives, for $0\le t\le T$,
\begin{align*}
\|\gamma_K(t,\cdot)\|_{L^2(0,X)}^2
&\le C_{3,s,T}\kappa_K^{-2(s-1)}
\int_0^t {\rm e}^{C_{1,s,T}(t-r)}\,{\rm d}r\\
&\le C_{3,s,T}T{\rm e}^{C_{1,s,T}T}
\kappa_K^{-2(s-1)}.
\end{align*}
Thus \eqref{eq:Galerkin-error-Hs} holds with, for example,
$
G_{s,T}:=C_{3,s,T}T{\rm e}^{C_{1,s,T}T}
$. The constants entering $G_{s,T}$ depend on $s$, $T$, $X$, the PDE
coefficients, and the uniform Sobolev bound $M_{s,T}$, but not on the
truncation order $K$.

%
\section{Proof of Proposition \ref{prop:physical-error-uniform-K}}
\label{ap:prop:physical-error-uniform-K}

The uniform bound \(\overline L = C_{\rm emb}M_{s,T}\) in Assumption~\ref{ass:uniform-physical-error-gain} relies on the Sobolev embedding
\(H^s_{\rm per}(0,X) \hookrightarrow W^{1,\infty}(0,X)\) for \(s > \frac{3}{2} \). 
For any \(w \in H^s_{\rm per}(0,X)\) with Fourier series $w(x) = \sum_{k \in \mbZ} \hat w_k {\rm e}^{i\omega_0 k x}$, an application of the triangle and Cauchy--Schwarz inequalities yields
\beqa
|w(x)| &\le& \frac{\|w\|_{H^s_{\rm per}(0,X)}}{\sqrt X} \Bigl(\sum_{k\in\mathbb Z}(1+\omega_0^2k^2)^{-s}\Bigr)^{\!1/2}
\quad \text{and} \nn\\
|\partial_x w(x)| &\le& \frac{\|w\|_{H^s_{\rm per}(0,X)}}{\sqrt X} \Bigl(\sum_{k\in\mathbb Z}\frac{\omega_0^2k^2}{(1+\omega_0^2k^2)^{s}}\Bigr)^{\!1/2},
\nn
\eeqa
where both series converge when \(s>\frac{3}{2}\). If we define the constants 
\[
C(s,X) := \Bigl(\sum_{k\in\mathbb Z}(1+\omega_0^2k^2)^{-s}\Bigr)^{\!1/2}<\infty
\quad \text{and}\quad
\tilde C(s,X) := \Bigl(\sum_{k\in\mathbb Z}\frac{\omega_0^2k^2}{(1+\omega_0^2k^2)^{s}}\Bigr)^{\!1/2}<\infty,
\]
then \(\|w\|_{L^\infty(0,X)} \le \frac{C(s,X)}{\sqrt{X}}\|w\|_{H^s_{\rm per}(0,X)}\) and
\(\|\partial_x w\|_{L^\infty(0,X)} \le \frac{\tilde C(s,X)}{\sqrt{X}}\|w\|_{H^s_{\rm per}(0,X)}\). Hence the norm of the embedding satisfies \(C_{\rm emb} \equiv C_{\rm emb}(s,X) \le  \frac{1}{\sqrt{X}}\bigl(C(s,X)+\tilde C(s,X)\bigr)\), which is finite when \(s>\frac{3}{2}\) and depends only on \(s\) and \(X\). Since \(M_{s,T}<\infty\) by Assumption~\ref{ass:uniform-physical-error-regularity}, it follows that \(\overline L = C_{\rm emb}M_{s,T} < \infty\), independently of \(K\), in Assumption \ref{ass:uniform-physical-error-gain}. 

Therefore, the Sobolev embedding together with Assumptions~\ref{ass:uniform-physical-error-regularity} and \ref{ass:uniform-physical-error-gain} yield
\beq
\sup_{K\ge K_0}\sup_{0\le t\le T}
\|\partial_xv_{{\rm ref},K}(t,\cdot)\|_{L^\infty(0,X)}
\le\overline L < \infty.
\label{eq-uniform-barL}
\eeq
Given Assumption \ref{ass:uniform-physical-error-regularity} and the uniform bound in \eqref{eq-uniform-barL}, Lemma~\ref{lem:global-dissipativity-fourier} holds for every $K \ge K_0$ with dissipativity constant $\mu_K \ge\mu_0$ and $\mu_0$ given by Assumption \ref{ass:uniform-physical-error-gain}. Specifically, taking $\Delta=-e_{{\rm ref},K}$ in \eqref{eq:global-dissipativity}, and recalling that ${\sf f}(t,0)=0$, Lemma~\ref{lem:global-dissipativity-fourier} yields
\beq
\|e_{{\rm ref},K}(t)\|_{K;L^2}^2
\le {\rm e}^{-2\mu_0t}\|e_{{\rm ref},K}(0)\|_{K;L^2}^2.
\label{eq-sost-1}
\eeq
Assumption~\ref{ass:uniform-physical-error-reference} gives \(\Delta_K(0)=0\), and Lemma~\ref{lem:global-dissipativity-fourier} gives \eqref{eq:ref-dissipativity} with \(\alpha=\mu_0\). The drift in \eqref{eq:Delta-SDE} is locally Lipschitz with coefficients bounded on \([0,T]\), and its diffusion is constant. Applying the localization argument from the proof of Proposition~\ref{prop:stochastic-closeness} (see expressions \eqref{eq---A}--\eqref{eq-indepe} and the ensuing non-explosion and Fatou arguments) with the gap $\alpha=\mu_0>0$ gives a unique non-explosive strong solution on \([0,T]\) and justifies taking expectations in the It\^o energy inequality \eqref{eq:Ito-Delta-est} with \(\alpha=\mu_0\). In this way we obtain
\[
\frac{{\rm d}}{{\rm d}t}\mbE\|\Delta_K(t)\|_{K;L^2}^2
\le -2\mu_0\mbE\|\Delta_K(t)\|_{K;L^2}^2
+\sigma^2d^2\|R_K\|_{K;L^2,F}^2
\]
for almost every \(t\in[0,T]\). Integrating with \(\Delta_K(0)=0\) and using \eqref{eq:uniform-reconstruction-bound} gives the uniform (over \(K\)) bound
\beq
\mbE\|\Delta_K(t)\|_{K;L^2}^2
\le\frac{\sigma^2d^2\overline R^2}{2\mu_0}
\left(1-{\rm e}^{-2\mu_0t}\right).
\label{eq-sost-2}
\eeq

Finally, we apply the inequality $\| \sum_{j=1}^4 z_j \|^2\le4\sum_{j=1}^4\|z_j\|^2$ to \eqref{eq:physical-error-full-decomposition}, then combine \eqref{eq:projection-tail-Hs}, \eqref{eq:Galerkin-error-Hs} in Lemma \ref{lem:uniform-Galerkin-error}, the bound $M_{s,T}$ in Assumption \ref{ass:uniform-physical-error-regularity} and the estimates \eqref{eq-sost-1} and \eqref{eq-sost-2} to prove \eqref{eq:physical-error-uniform-K}. For any fixed $t<\infty$, choose $T\ge t$. The constants $M_{s,T}$ and $G_{s,T}$ are then finite and independent of $K$ (see Assumption \ref{ass:uniform-physical-error-regularity} and \ref{lem:uniform-Galerkin-error}), so the projection and Galerkin terms vanish as $K\to\infty$, which proves \eqref{eq:physical-error-K-limit}.  If Assumption~\ref{ass:uniform-physical-error} holds uniformly on $[0,\infty)$, then $\mu_0$ is independent of $t$; taking $\limsup_{t\to\infty}$ in \eqref{eq:physical-error-K-limit} proves \eqref{eq:physical-error-time-K-limit} ($G_{s,T}$ does not need to be bounded uniformly in $T$ because the limit $K\to\infty$ is taken first).

\begin{singlespace}
\fontsize{10pt}{12pt}\selectfont
\bibliographystyle{plain}
\bibliography{references_v11}
\end{singlespace}

\end{document}